\documentclass[journal]{IEEEtran}
\usepackage{amsmath}
\usepackage{mathrsfs}
\usepackage{latexsym}
\usepackage{amssymb}
\usepackage{amsfonts}
\usepackage{pifont}
\usepackage{epsfig}
\usepackage{theorem}
\usepackage{array}

\def\Pr{{\mathbb P}}

\def\eqdef{\triangleq}

\def\nat{{\mathbb  N}}
\def\eqd{\stackrel{d}{=}}

\def\ind{{\bf 1}}
\def\leqd{\stackrel{d}{\le}}
\def\geqd{\stackrel{d}{\ge}}
\def\12{\frac{1}{2}}

\def\expect{{\mathbb  E}}
\def\Pr{{\mathbb P}}

\def\nat{{\mathbb  N}}

\def\eqdef{\triangleq}
\def\eqd{\stackrel{d}{=}}
\def\leqd{\stackrel{d}{\le}}
\def\ind{{\bf 1}}

\def\bar{\overline}

\newtheorem{theorem}{Theorem}[section]
\newtheorem{proposition}{Proposition}[section]
\newtheorem{lemma}{Lemma}[section]
\newtheorem{corollary}{Corollary}[section]

{\theorembodyfont{\rmfamily}
\newtheorem{example}{Example}
\newtheorem{remark}{Remark}
}

\ifCLASSINFOpdf
\else
\fi
\begin{document}
%
% paper title
% can use linebreaks \\ within to get better formatting as desired
\title{Stability of Finite Population ALOHA with Variable Packets}
%
%
% author names and IEEE memberships
% note positions of commas and nonbreaking spaces ( ~ ) LaTeX will not break
% a structure at a ~ so this keeps an author's name from being broken across
% two lines.
% use \thanks{} to gain access to the first footnote area
% a separate \thanks must be used for each paragraph as LaTeX2e's \thanks
% was not built to handle multiple paragraphs
%

\author{Predrag R. Jelenkovi\'c and Jian Tan
\thanks{Predrag R. Jelenkovi\'c and Jian Tan are with the Department
of Electrical Engineering, Columbia University, New York, 10027 USA,
 e-mail: \{predrag,jiantan\}@ee.columbia.edu.}% <-this % stops a space
 \thanks{The preliminary version of
this paper has appeared in \cite{Jelen06ALOHA}.} }

\maketitle

\begin{abstract}
ALOHA is one of the most basic Medium Access Control (MAC) protocols
and represents a foundation for other more sophisticated distributed
and asynchronous MAC protocols, e.g., CSMA.   In this paper, unlike
in the traditional work that focused on mean value analysis, we
study the distributional properties of packet transmission delays
over an ALOHA channel. We discover a new phenomenon showing that a
basic finite population ALOHA model with variable size (exponential)
packets is characterized by power law transmission delays, possibly
even resulting in zero throughput. These results are in contrast to
the classical work that shows exponential delays and positive
throughput for finite population ALOHA with fixed packets.
Furthermore, we characterize a new stability condition that is
entirely derived from the tail behavior of the packet and backoff
distributions that may not be determined by mean values.  The power
law effects and the possible instability might be diminished, or
perhaps eliminated, by reducing the variability of packets. However,
we show that even a slotted (synchronized) ALOHA with packets of
constant size can exhibit power law delays when the number of active
users is random. From an engineering perspective, our results imply
that the variability of packet sizes and number of active users need
to be taken into consideration when designing robust MAC protocols,
especially for ad-hoc/sensor networks where other factors, such as
link failures and mobility,  might further compound the problem.
\end{abstract}
% IEEEtran.cls defaults to using nonbold math in the Abstract.
% This preserves the distinction between vectors and scalars. However,
% if the journal you are submitting to favors bold math in the abstract,
% then you can use LaTeX's standard command \boldmath at the very start
% of the abstract to achieve this. Many IEEE journals frown on math
% in the abstract anyway.

% Note that keywords are not normally used for peerreview papers.

% For peer review papers, you can put extra information on the cover
% page as needed:
% \ifCLASSOPTIONpeerreview
% \begin{center} \bfseries EDICS Category: 3-BBND \end{center}
% \fi
%
% For peerreview papers, this IEEEtran command inserts a page break and
% creates the second title. It will be ignored for other modes.
\IEEEpeerreviewmaketitle

\section{Introduction}
%\hspace{-1.7mm}\footnote{Published in Proceedings of ITC-20, Ottawa,
%Canada, June 17-21, 2007.}
 ALOHA
represents one of the first and most basic distributed Medium Access
Control (MAC) protocols \cite{Abr70}. It is easy to implement since
it does not require any user coordination or complicated controls
and, thus, represents a basis for many modern MAC protocols, e.g.,
Carrier Sense Multiple Access (CSMA). Basically, ALOHA  enables
multiple users to share a common communication medium (channel) in a
completely uncoordinated manner. Namely, a user attempts to send a
packet over the common channel and, if there are no other user
(packet) transmissions during the same time, the packet is
considered successfully transmitted. Otherwise, if the transmissions
of more than one packet (user) overlap, we say that there is a
collision and the colliding packets need to be retransmitted. Each
user retransmits a packet after waiting for an independent (usually
exponential/geometric) period of time, making ALOHA entirely
decentralized and asynchronous. The desirable properties of ALOHA,
including its low complexity and distributed/asynchronous nature,
make it especially beneficial for wireless sensor networks with
limited resources as well as for wireless ad hoc networks that have
difficulty in carrier sensing due to hidden terminal problems and
mobility. Furthermore, because of these properties ALOHA represents
a basis for many more sophisticated MAC protocols, e.g., CSMA.

Traditionally, the performance evaluation of ALOHA has focused on
mean value (throughput) analysis, the examples of which can be found
in every standard textbook on networking, e.g., see
\cite{BG92,Sch86,RS90}; for more recent references see \cite{NMT03}
and the references therein (due to space limitations, we do not
provide comprehensive literature review on ALOHA in this paper).
However, it appears that there are no explicit and general studies
(more than two users) of the distributional properties of ALOHA,
e.g., delay distributions. In this regard, in
Subsection~\ref{sec:description}, we consider a standard finite
population ALOHA model with variable length packets
\cite{Fer77,BB80} that have an asymptotically exponential tail.
Surprisingly, we discover a new phenomenon that the distribution of
the number of retransmissions (collisions) and time between two
successful transmissions follow power law distributions, as stated
in Proposition~\ref{theorem:preliminary} of Subsection \ref{ss:pre},
Theorem \ref{theorem:transient} of Subsection~\ref{ss:starting} on
starting behavior as well as Theorem~\ref{theorem:steadyLower} of
Subsection~\ref{ss:steadyState} on steady state behavior.  Based on
this observation, we derive new stability conditions for finite
population ALOHA with variable packets
%Corollaries \ref{cor:pre1}
%and \ref{cor:pre2} of Subsection~\ref{ss:pre} as well as a refined
%one
in Theorem~\ref{theorem:refinedStability} of
Section~\ref{section:stability}.
% and this
%power law phenomenon is described by an interesting bifurcation of
%the delay distribution from a single power law, after finitely many
%transmissions, to multiple power law distributions in steady state,
%as stated in Theorem~\ref{theorem:transient} of
%Subsection~\ref{ss:starting} on starting behavior and
%Theorems~\ref{theorem:steadyUpper} and \ref{theorem:steadyLower} of
%Subsection~\ref{ss:steadyState} on steady state behavior.
Informally, our theorem shows that when the exponential decay rate
of the packet distribution is smaller than the parameter of the
exponential backoff distribution and the arrival rate, even the
finite population ALOHA may have zero throughput. This is contrary
to the common belief that the finite population ALOHA system always
has a positive, albeit possibly small, throughput. Furthermore, even
when the long term throughput is positive, the high variability of
power laws (infinite variance when the power law exponent is less
than $2$) may cause unstable buffer content (queue sizes), implying
 periods of very high congestion, long delays, and low throughput.
It also may appear counterintuitive that the system is characterized
by power laws even though the distributions of all the variables
(arrivals, backoffs and packets) of the system are of exponential
type. However, this is in line with the results in
\cite{FS05,SL06,PJ06RETRANS}, which show that job completion times
in systems with failures where jobs restart from the very beginning
exhibit similar power law behavior. Our study in \cite{PJ06RETRANS}
was done in the communication context where job completion times are
represented by document/packet transmission delays, e.g., ARQ
protocol. It may also be worth noting that \cite{PJ06RETRANS}
reveals the existence of power law delays regardless of how light or
heavy the packet/document and link failure distributions may be
(e.g., Gaussian), as long as they have proportional hazard
functions. Furthermore, from a mathematical perspective,
Proposition~\ref{theorem:steadyUpper},
Theorems~\ref{theorem:transient}
 and \ref{theorem:steadyLower} analyze a
more complex setting than the one in  \cite{PJ06RETRANS,SL06} and,
thus, require a novel proof. Hence, when compared with
\cite{PJ06RETRANS,SL06}, this paper both discovers a new related
phenomenon in a communication MAC layer application area and
provides a novel analysis of it.

As already stated in the abstract, the preceding power law
phenomenon is a result of combined effects of packet variability and
collisions. Hence, one can see easily that the power law delays can
be eliminated by reducing the variability of packets. Indeed, for
slotted ALOHA with constant size packets the delays are
geometrically distributed. However, we show in
Section~\ref{s:slotted} that, when the number of users sharing the
channel is geometrically distributed, the slotted ALOHA exhibits
power law delays as well.

In Section~\ref{s:NSE}, we illustrate our results with simulation
experiments, which show that the asymptotic power law regime is
valid even for relatively small delays and reasonably large
probability values. Furthermore, the distribution of packets/number
of users in practice might have a bounded support. To this end, we
show by a simulation experiment that this situation results  in
distributions that have power law main body with an exponentiated
(stretched) support in relation to the support of the packet
size/number of active users. Hence, although exponentially bounded,
the delays may be prohibitively long.

In practical applications, we may have combined effects of both
variable packets and a random number of users, implying that the
delay and congestion is likely to be even worse than predicted by
our results. Thus, from an engineering perspective, one has to pay
special attention to the packet variability and the number of users
when designing robust MAC protocols, especially for ad-hoc/sensor
networks where link failures \cite{PJ06RETRANS}, mobility and many
other factors might further worsen the performance.

In summary, the rest of the paper is organized as follows.  In
Section~\ref{s:finitePop}, we provide the description and the
preliminary power law bounds. Then, we present our new stability
conditions that are based on packet distribution decay rates in
Section~\ref{section:stability}. Further distributional properties
for the number of retransmissions and delays are investigated in
Section~\ref{section:distributional}. Section~\ref{s:slotted}
contains the results on power laws in slotted ALOHA with random
number of users. Experimental validation of our results can be found
in Section~\ref{s:NSE}. The paper is concluded in
Section~\ref{s:alohaconcluding}. Finally, some of the more technical
proofs are postponed to Section~\ref{sc:ALOHAproof}.

\section{Power Laws in the Finite Population ALOHA with Variable Size Packets}\label{s:finitePop}
In this section we show that the variability of packet sizes,  when
coupled with the contention nature of ALOHA, is a cause of power law
delays. This study is motivated by the well-known fact that packets
in today's Internet have variable sizes. To further emphasize that
packet variability is a sole cause of power laws, we assume a finite
population ALOHA model where each user can hold (queue) up to one
packet at a time since the increased queueing only further
exacerbates the problem. In addition, in Section~\ref{s:slotted} we
show that the user variability in an infinite population model may
be a cause of power law delays as well. In the remainder of this
section, we describe the model and introduce the necessary notation
in Subsection~\ref{sec:description} and  present the preliminary
results in Subsection~\ref{ss:pre}.
% In
%Section~\ref{section:stability} and
%Section~\ref{section:distributional}  we formulate and prove our
%main results on the logarithmic asymptotics of the transmission
%delay for the starting behavior in Theorem~\ref{theorem:transient}
%and for the steady behavior in Theorems~\ref{theorem:steadyUpper}
%and \ref{theorem:steadyLower}.

\subsection{Model Description}\label{sec:description}
Consider $M\ge 2$ users sharing a common communication link
(channel) of unit capacity. Each user can hold at most one packet in
its queue and, when the queue is empty, a new packet is generated
after an independent (from all other variables) exponential time
with mean $1/\lambda$. Each packet has an independent length that is
equal in distribution to a generic random variable $L$. A user with
a newly generated packet attempts its transmission immediately and,
if there are no other users transmitting during the same time, the
packet is considered successfully transmitted. Otherwise, if the
transmissions of more than one packet overlap, we say that there is
a collision and the colliding packets need to be retransmitted; for
a visual representation of the system see Figure~\ref{fig:aloha}.
After a collision, each participating user waits (backoffs) for an
independent exponential period of time with mean $1/\nu$ and then
attempts to retransmit its packet.  Each such user continues this
procedure until its packet is successfully transmitted and then it
generates a new packet after an independent exponential time of mean
$1/\lambda$.  Let $\{U(t)\}_{t\ge 0}$ denote the number of users
that are in backoff state at time $t$ and $\left\{L_i^{(t)}
\right\}_{1\le i\le U(t)}$ denote the packet sizes of all the $U(t)$
number of active users at time $t$.

 \begin{figure}[h]
\centering
\includegraphics[width=8.8cm]{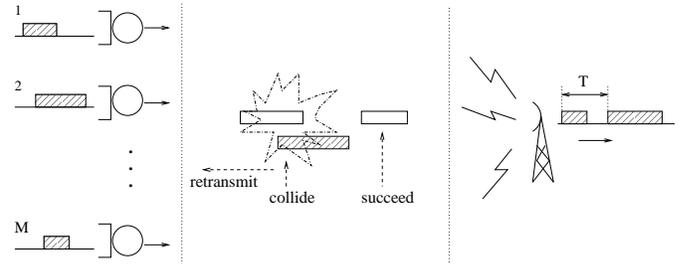}
  \caption{Finite population ALOHA model with variable packet sizes.}
\label{fig:aloha}
\end{figure}

%Without loss of generality, assume that there is a successful
%transmission at time $t=0$ and
From the perspective of the receiver, let $\{C_i\}_{i\ge 1}$ be an
increasing sequence of positive time points when either a collision
or successful transmission occurs with $C_0=0$. Let $\{D_m\}_{m\ge
1}$ be the sequence of time points when the receiver successfully
receives the $m$th packet and define $T_m=D_{m}-D_{m-1}$ to be the
transmission time for the $m$th successfully received packet with a
convention $D_0 = 0$. Correspondingly, we can define $N_m$ to be the
number of (re)transmissions in the interval $(D_{m-1}, D_{m}]$ for
the $m$th successful transmission.

Now, from the perspective of user $i, , 1\le i \le M$, define
$\left\{D_m^{(i)}\right\}_{m\ge 1}$ to be the sequence of time
points when user $i$ successfully sends the $m$th packet and define
$T_m^{(i)}=D_{m}^{(i)}-D_{m-1}^{(i)}$ to be the transmission time
for the $m$th successfully transmitted packet with a convention
$D_0^{(i)} = 0$. By the same fashion, we can define $N_m^{(i)}$ to
be the number of (re)transmissions in the interval $(D_{m-1}^{(i)},
D_{m}^{(i)}]$ for the $m$th successful transmission.

 We will study the stability of this model as well as the asymptotic properties of
the distributions of $N_m, T_m$ and $N_m^{(i)}, T_m^{(i)}$.

%Variable $T$ has the following useful representation that is implied by the memoryless property of
%exponential packet generation and backoff times.
%To this end, let $U(t), t\ge 0$ be the number packets in the system (held by all users) at time $t$ and
%let $\{Y_i(\lambda_i)\}_{i\geq 1}$ be a sequence of independent
%exponential random variables each having mean $\lambda_i$
%respectively.  Then, $T$ can be represented in distribution by
%\begin{equation*}
% T\eqd \sum_{i=1}^{N}Y_i\left( (M-U(T_{i-1})) \lambda +  U(T_{i-1})\mu\right),
%\end{equation*}
%where $T_0=0$.
\subsection{Power Law Bounds}\label{ss:pre}
In the rest of this subsection,  we present preliminary results for
the finite ALOHA with variable packets,  described in the preceding
subsection.  Let $x \wedge y=\min(x, y)$, $x\vee y = \max(x, y)$,
and $\leqd $, $\geqd$, $\eqd$ denote inequalities and equality in
distribution, respectively.

Basically, ALOHA model can be viewed as a state dependent channel
with failures where the failure rate depends on the number of
backoffed users and the sizes of the packets present in the system.
Hence, this model can be viewed as a generalization of the problem
stated in \cite{PJ06RETRANS,Jelen07}.  The following proposition
shows that the distributions of the number of retransmissions and
the delays in our ALOHA model are always sandwiched between two
power laws, which is obtained by uniformly bounding the variable
collision (failure) rates independently of the state of the channel.
\begin{proposition}\label{theorem:preliminary}
Assume that , for $\mu>0$,
\begin{equation}\label{eq:condition}
   \lim_{x\rightarrow \infty}\frac{\log \Pr[L>x]}{x}=-\mu,
   \end{equation}
   and let $\underline{N}$ and $\overline{N}$ be two random
   variables with distributions
   \begin{equation}
  \Pr\left[\underline{N}>n\right] = \expect \left[\left( 1- e^{-L (M-1)(\lambda \wedge \nu)}
 \right)^n \right]
   \end{equation}
   and
   \begin{equation}
  \Pr\left[\overline{N}>n\right] = \expect \left[\left( 1- \frac{\lambda \wedge \nu}{ M
(\lambda \vee \nu)} e^{-L (M-1)(\lambda \vee \nu)}
 \right)^n \right]. \nonumber
   \end{equation}
 Then, uniformly for all $m$ and $i$,
   \begin{equation} \label{eq:presanwich}
\underline{N} \leqd N_m^{(i)} \leqd \overline{N}
    \end{equation}
and
% Then, for all $1\leq i \leq M$, there exist
%$\underline{N}^{(i)} , \bar{N}^{(i)} $ and $\underline{T}^{(i)} ,
%\bar{T}^{(i)} $, such that $\underline{N}^{(i)}\leq N_m^{(i)} \leq
%\bar{N}^{(i)}, \underline{T}^{(i)} \leq T_m^{(i)} \leq \bar{T}^{(i)}
%$, and
\begin{equation}\label{eq:lowboundNm}
     \lim_{n\rightarrow \infty}\frac{\log \Pr\left[\underline{N}>n\right]}{\log
    n} = -\frac{\mu}{(M-1)(\lambda \wedge \nu )},
   \end{equation}
   \begin{equation}\label{eq:upboundNm}
\lim_{n\rightarrow \infty}\frac{\log
\Pr\left[\overline{N}>n\right]}{\log
    n}=  -\frac{\mu}{(M-1)(\lambda \vee \nu )}.
   \end{equation}
   Similarly, there exist $\underline{T}$ and $\overline{T}$ such
   that (\ref{eq:presanwich}), (\ref{eq:lowboundNm}) and
   (\ref{eq:upboundNm}) are satisfied for the corresponding
   expressions for $T_m^{(i)}$, $\underline{T}$ and $\overline{T}$
   (replacing $N$ by $T$).
%   \begin{equation}\label{eq:lowboundTm}
%     \varliminf_{t\rightarrow \infty}\frac{\log \Pr\left[\underline{T}^{(i)}>t\right]}{\log
%    t} \geq -\frac{\mu}{(M-1)(\lambda \wedge \nu )},
%   \end{equation}
%   \begin{equation}\label{eq:upboundTm}
%\varlimsup_{t\rightarrow \infty}\frac{\log
%\Pr\left[\bar{T}^{(i)}>t\right]}{\log
%    t} \leq  -\frac{\mu}{(M-1)(\lambda \vee \nu )}.
%   \end{equation}
\end{proposition}
\begin{proof}
We begin with studying $N_m^{(i)}$. First, we prove the \emph{lower
bound}. Note that a collision for user $i$ may occur for two
different reasons. Either, when user $i$ attempts to access the
channel it collides with the already existing transmission, or after
user $i$ successfully starts its transmission it is interrupted
later by some other user that tries to access the channel. Now, if
$\underline{N}^{(i)}_m$ only counts the collisions due to the second
reason, then clearly $N_m^{(i)} \geq \underline{N}^{(i)}_m$.
Similarly, if $\underline{T}^{(i)}_m$ is the total time that only
measures the delay caused by the collisions of the second type, then
$T_m^{(i)} \geq \underline{T}^{(i)}_m$.

Now, consider the system at the moment when user $i$ has
successfully initiated its transmission. At that moment a number of
users ($\le M-1$) can be in the backoffed state (exponential with
rate $\nu$ for each user) and the remaining ones are waiting for the
new packets to arrive (exponential with rate $\lambda$ for each
user). Hence, the time until another user attempts to access the
channel is upper bounded by an exponential time of rate
$(M-1)(\lambda \wedge \nu)$. Therefore, given $L$, the probability
that there is a collision of the second type is lower bounded by
$1-e^{-L(M-1)(\lambda \wedge \nu)}$, implying that
\begin{align}\label{eq:templateLowbound5}
 \Pr\left[N_m^{(i)} > n\right] &\geq \Pr[\underline{N}^{(i)}_m>n
 ]\nonumber\\
 &
 \geq \expect \left[\left( 1- e^{-L (M-1)(\lambda \wedge \nu)}
 \right)^n \right]\nonumber\\
 &=  \Pr[\underline{N}>n ],
\end{align}
since the repetitions of exponential times of rate $(M-1)(\lambda
\wedge \nu)$ are independent due to the memoryless property.

 Condition
(\ref{eq:condition}) implies that, for any $\epsilon>0$, there
exists $x_{\epsilon}$ such that $\Pr[L>x] \geq e^{-(\mu+\epsilon)x}$
for
 all $x\geq  x_{\epsilon} $.
Then, if we define random variable $L_{\epsilon}$
 with $\Pr[L_{\epsilon}>x]=e^{-(\mu+\epsilon)x}, x\geq 0$,
we obtain
$$
    L \geqd L_{\epsilon} \ind(L_{\epsilon}>x_{\epsilon}),
$$
resulting in
\begin{align}\label{eq:templateLowbound}
    \Pr\left[\underline{N} > n\right] & \geq \expect \left[\left(1-e^{-L_{\epsilon}
    \ind(L_{\epsilon}>x_{\epsilon})(M-1)(\lambda \wedge \nu)} \right)^n
  \right]\nonumber \\
& \geq \expect \left[\left(1-e^{-L_{\epsilon}
 (M-1)(\lambda \wedge \nu)} \right)^n \ind(L_{\epsilon}>x_{\epsilon})
  \right].
  \end{align}
Noticing that for any $\delta>0$, there exists $0<x_{\delta}<1$ such
that $1-x\geq e^{-(1+\delta)x}$ for all $0\le x \le x_{\delta}$, we
can choose $x_{\epsilon}$ large enough, such that
  \begin{align}\label{eq:preLow}
 \Pr\left[\underline{N} > n\right] & \geq \expect \left[ e^{-(1+\epsilon)n e^{-L_{\epsilon}
 (M-1)(\lambda \wedge \nu)} } \ind(L_{\epsilon}>x_{\epsilon})
  \right]\nonumber \\
  & = \expect \left[ e^{-(1+\epsilon)n e^{-L_{\epsilon}
 (M-1)(\lambda \wedge \nu)} } \right] \nonumber\\
 &\quad \quad -\expect \left[ e^{-(1+\epsilon)n e^{-L_{\epsilon}
 (M-1)(\lambda \wedge \nu)} } \ind(L_{\epsilon}\leq x_{\epsilon})
  \right]\nonumber\\
  &\geq \expect \left[ e^{-(1+\epsilon)n e^{-L_{\epsilon}
 (M-1)(\lambda \wedge \nu)} } \right]- \zeta^n,
  \end{align}
  where $\zeta = e^{-(1+\epsilon)e^{-x_{\epsilon}
 (M-1)(\lambda \wedge \nu}) }<1$.

Now, for any $0<x<1$,
\[
\Pr\left[e^{-(\mu+\epsilon)L_{\epsilon}}<x\right]=\Pr
\left[L_{\epsilon}>- \frac{\log x }{ \mu+\epsilon} \right]=x,
\]
implying that $e^{-(\mu+\epsilon)L_{\epsilon}} \eqd U$, where  $U$
is a uniform random variable between $0$ and $1$.
  Thus, we can derive from (\ref{eq:preLow})
    \begin{align*}
  \Pr[\underline{N} > n] &\geq
  \expect \left[e^{-n (1+\epsilon) U^{(M-1)(\lambda \wedge \nu)/(\mu+\epsilon)}}
   \right] - \zeta^n.
\end{align*}
Now, since
\begin{equation}\label{eq:unif2power}
\expect[e^{-\theta U^{1/\alpha}}]\sim
\Gamma(\alpha+1)/\theta^{\alpha}
\end{equation}
 as $\theta \to \infty$,  one can
easily obtain
\begin{equation}\label{eq:upperN1}
\varliminf_{n\rightarrow \infty}\frac{\log \Pr\left[\underline{N}
> n\right]}{\log n} \geq -\frac{\mu+\epsilon}{(M-1)(\lambda \wedge
\nu)},
\end{equation}
which,  by passing $\epsilon \to 0$,  proves the lower bound.

Next, we prove the \emph{upper bound}. We observe that a successful
transmission has two steps. First, the user has to initiate the
transmission successfully (grab the channel). Second, after
accessing the channel it has to complete the transmission without
interruptions from other users.  We will  bound these events by
independent ones as described below.

After a successful transmission or a collision, user $i$ will
attempt to access the channel after an exponential time of rate no
smaller than $\lambda \wedge \nu$;  each other user will compete to
access the channel after an exponential time of rate no larger than
$\lambda \vee \nu$. Once user $i$ grabs the channel, its
transmission will be successful if the first channel access time of
all the other users is larger than $L$. Note that the first access
time of the other users is exponential with rate upper bounded by
$(M-1)(\lambda \vee \nu)$. Therefore,  given $L$, the probability of
a collision (failure) is upper bounded by
$$
1-  \frac{\lambda \wedge \nu}{ M (\lambda \vee \nu)}
e^{-L(M-1)(\lambda \vee
  \nu)}.
$$
Furthermore, due to the memoryless property of exponential
distribution the probability of $n$ successive collisions, given
$L$, can be upper bounded by independent events with probabilities
given by the preceding expression. Therefore, after unconditioning,
we obtain
\begin{align}\label{eq:preliminaryUp}
 \Pr\left[N_m^{(i)} >
n\right] &\leq  \expect \left[\left( 1- \frac{\lambda \wedge \nu}{ M
(\lambda \vee \nu)} e^{-L (M-1)(\lambda \vee \nu)}
 \right)^n \right] \nonumber\\
 &= \Pr\left[\bar{N}>n\right].
\end{align}

Next, using $1-x\leq e^{-x}$ and defining $\zeta \eqdef \lambda
\wedge \nu / (M (\lambda \vee \nu))$, we derive, for
$x_{\epsilon}>0$,
\begin{align}\label{eq:hi1a}
 \Pr\left[\bar{N}>n\right] &\leq  \expect \left[ e^{-n\zeta e^{-L(M-1)(\lambda \vee
\nu) }}
 \ind\left(L> x_{\epsilon}\right) \right] \nonumber\\
 &\quad \quad +
  \left(1-\zeta e^{-x_{\epsilon}(M-1)(\lambda\vee \nu)} \right)^{n} \nonumber\\
  &\leq  \expect \left[ e^{-n\zeta e^{-L\ind\left(L> x_{\epsilon}\right) (M-1)(\lambda \vee \nu) }}
 \right] + \eta^n ,
\end{align}
where $\eta\eqdef 1-\zeta e^{-x_{\epsilon}(M-1)(\lambda\vee
\nu)}<1$. Now, condition (\ref{eq:condition}) implies that,   for
any $0<\epsilon<\mu$, we can choose $x_{\epsilon}$ such that
$\Pr[L>x] \leq e^{-(\mu-\epsilon)x}$ for all $x\geq x_{\epsilon}$.
Thus,  by defining an exponential random variable $L^{\epsilon}$
with $\Pr[L^{\epsilon}>x]=e^{-(\mu-\epsilon)x}, x\geq 0$, we obtain
$
   L \ind(L>x_{\epsilon}) \leqd L^{\epsilon}
$. Therefore, (\ref{eq:hi1a}) implies
\begin{align}\label{eq:hi2a}
  \Pr\left[\bar{N}>n\right]
 &\leq  \expect \left[ e^{-n \zeta e^{- L^{\epsilon}(M-1)(\lambda \vee \nu)}}
 \right] + \eta^n.
\end{align}
Similarly as in the proof of the lower bound, we know
$e^{-(\mu-\epsilon)L^{\epsilon}} \eqd U$ is a uniform random
variable between $0$ and $1$.
  Thus, (\ref{eq:hi2a}) implies
  \begin{align} %\label{eq:prelim}
 \Pr\left[\bar{N}>n\right] &\leq
  \expect \left[e^{-n \zeta U^{(M-1)(\lambda \vee \nu)/(\mu-\epsilon)}}
   \right] +
  \eta^{n}. \nonumber
\end{align}
By (\ref{eq:unif2power}), we obtain
\begin{equation}
\varlimsup_{n\rightarrow \infty}\frac{\log
\Pr\left[\bar{N}>n\right]}{\log n} \leq
-\frac{\mu-\epsilon}{(M-1)(\lambda \vee \nu)}, \nonumber
\end{equation}
which, by passing $\epsilon \to 0$,  finishes the proof of the upper
bound.
% Note that the right hand side of the inequality
%(\ref{eq:prelim}) does not depend on $m$,  hence there exists
%$\bar{N}^{(i)}$ such that $N_m^{(i)}
% \leqd \bar{N}^{(i)}$ and
%\begin{equation}
%\varlimsup_{n\rightarrow \infty}\frac{\log \Pr\left[\bar{N}^{(i)}>
%n\right]}{\log n} \leq -\frac{\mu}{(M-1)(\lambda \vee \nu)}.
%\nonumber
%\end{equation}

Now, we prove the result for $T_m^{(i)}$.  Observe that each attempt
for user $i$ to transmit the $m$th packet consists of two steps.
First, user $i$ initiates an attempt to grab the channel; for the
$j$th attempt, denote by $\{X_j\}_{j\ge 1}$ the idle period where
user $i$ either is waiting for a new packet to arrive or is in its
backoff state for the $m$th packet. Hence, $X_1$ is exponential with
rate $\lambda$ and $X_j, j>1$ are exponential with rate $\nu$.
Second, after user $i$ makes an attempt to access the channel, it
either collides with other users that are transmitting packets or
starts transmitting its own packet; for the $j$th attempt, denote by
$\{Y_j\}_{j\ge 1}$ the period during which there are no
transmissions from other users after user $i$ starts sending the
$m$th packet. Note that if user $i$ fails to grab the channel for
the $j$th attempt, then $Y_j=0$; if user $i$ successfully grabs the
channel for this attempt, then it spends time $Y_j$ transmitting the
$m$th packet without interference from other users. Thus, we have
\begin{equation}\label{eq:preTmRep}
     T_m^{(i)} = \sum_{j=1}^{N_m^{(i)}}X_j+\sum_{j=1}^{N_m^{(i)}-1}Y_j
     +L.
\end{equation}
Since $\{X_j\}$ is a sequence of exponential random variables with
rate equal to either  $\lambda$ or $\nu$,  we can always find two
i.i.d. exponential sequences, $\{\underline{X},
\underline{X}_j\}_{j\ge 1}$ and  $\{\bar{X}, \bar{X}_j\}_{j\ge 1}$,
such that
\begin{equation}\label{eq:preTmRep1}
\underline{X}_j \leq X_j \leq \bar{X}_j.
\end{equation}
 Additionally, observe
that when user $i$ successfully grabs the channel, $Y_j$ is
stochastically upper bounded by an exponential random variable with
rate $(M-1)(\lambda \vee \nu)$, and thus, we can construct a
sequence of i.i.d. exponential random variables $\{\bar{Y},
\bar{Y}_j\}$ such that
\begin{equation}\label{eq:preTmRep2}
Y_j \leq \bar{Y}_j,
\end{equation}
where $\{\bar{Y}_j\}$ is independent of $\{\bar{X}_j\}$.

First, we prove the upper bound.  Using the union bound,
\begin{align}
  \Pr\left[T_m^{(i)}>2t\right] &\leq \Pr\left[\sum_{j=1}^{N_m^{(i)}}({X}_j+{Y}_j)+L>2t\right] \nonumber\\
  &\leq \Pr\left[\sum_{j=1}^{N_m^{(i)}}({X}_j+{Y}_j)>t, N_m^{(i)}\leq
  \frac{t}{2\expect[\bar{X}+\bar{Y}]}\right]\nonumber\\
  &\quad    + \Pr\left[N_m^{(i)}>
  \frac{t}{2\expect[\bar{X}+\bar{Y}]}\right]+\Pr\left[L >t\right] \nonumber\\
  &\leq
  \Pr\left[\sum_{j=1}^{t/\left(2\expect[\bar{X}+\bar{Y}]\right)}({X}_j+{Y}_j)>t\right]\nonumber\\
  &\quad
   + \Pr\left[N_m^{(i)}>
  \frac{t}{2\expect[\bar{X}+\bar{Y}]}\right]+\Pr\left[L
  >t\right]. \nonumber
  \end{align}
By  (\ref{eq:presanwich}), (\ref{eq:preTmRep1}) and
(\ref{eq:preTmRep2}), we obtain
  \begin{align}\label{eq:upperT3}
   \Pr\left[T_m^{(i)}>2t\right]  &\leq
   \Pr\left[\sum_{j=1}^{t/\left(2\expect[\bar{X}+\bar{Y}]\right)}(\bar{X}_j+\bar{Y}_j)>t\right]
   \nonumber\\
   &
   + \Pr\left[\bar{N}>
  \frac{t}{2\expect[\bar{X}+\bar{Y}]}\right]+\Pr\left[L >t\right]
  \nonumber\\
  &\eqdef I_1+I_2+I_3,
  \end{align}
which,  by defining random variable $\bar{T}$
  with the following distribution
  \begin{align}
  \Pr\left[\bar{T}>2t\right]
&\eqdef \min\{I_1+I_2+I_3, 1\}, \nonumber
\end{align}
implies $\Pr\left[T_m^{(i)}>2t\right] \leq
   \Pr\left[\bar{T}>2t\right]$, i.e.,
 \begin{align} \label{eq:upperT4}
T_m^{(i)} \leqd \bar{T}.
\end{align}

For (\ref{eq:upperT3}), applying Chernoff bound, we derive
$I_1=O(e^{-\eta n})$ for some
 $\eta>0$.  Condition (\ref{eq:condition}) implies $I_3=O(e^{-\eta n})$ for
 some other
 $\eta>0$.   To compute $I_2$,  using (\ref{eq:upboundNm}), we
 obtain,
$$
\lim_{t\rightarrow \infty}\frac{\log \Pr\left[\bar{N}>
  \frac{t}{2\expect[\bar{X}+\bar{Y}]}\right]}{\log t} = -\frac{\mu}{(M-1)(\lambda \vee
  \nu)},
$$
 which,
 combined with the estimates for $I_1$ and $I_3$, implies that
\begin{equation}\label{eq:preTmUp2}
\lim_{t\rightarrow \infty}\frac{\log \Pr\left[\bar{T} >t
\right]}{\log t} = -\frac{\mu}{(M-1)(\lambda \vee \nu)}.
\end{equation}

Next, we prove the lower bound. It is easy to obtain
\begin{align}
  \Pr\left[T_m^{(i)} > t \right] &\geq
  \Pr\left[\sum_{j=1}^{N_m^{(i)}-1}(X_j+Y_j)+L>t\right]\nonumber\\
  &
  \geq \Pr\left[\sum_{j=1}^{N_m^{(i)}-1}X_j>t\right] \nonumber\\
  &\geq \Pr\left[\sum_{j=1}^{N_m^{(i)}-1}X_j>t,
     N_m^{(i)} > \frac{2t}{\expect[\underline{X}]}+1\right]\nonumber\\
       &\geq \Pr\left[N_m^{(i)} >
       \frac{2t}{\expect[\underline{X}]}+1\right] \nonumber\\
       &\quad\quad
    -\Pr\left[\sum_{j=1}^{N_m^{(i)}-1}X_j\leq t,
     N_m^{(i)} \geq \frac{2t}{\expect[\underline{X}]}+1\right], \nonumber
\end{align}
which, by recalling $X_j \geq \underline{X}_j$ and using
(\ref{eq:presanwich}), yields
\begin{align} %\label{eq:lowerT1}
   \Pr\left[T_m^{(i)} > t\right]&\geq \Pr\left[N_m^{(i)} > \frac{2t}{\expect[\underline{X}]}+1\right]
    -\Pr\left[\sum_{j=1}^{2t/\expect[\underline{X}]}X_j\leq t
    \right] \nonumber\\
   &\geq \Pr\left[\underline{N} > \frac{2t}{\expect[\underline{X}]}+1\right]
    -\Pr\left[\sum_{j=1}^{2t/\expect[\underline{X}]}\underline{X}_j\leq t
    \right] \nonumber\\
   &\eqdef I_1(t)-I_2(t). \nonumber
\end{align}

Now, define a random variable $\underline{T}$ with
\begin{align}
   \Pr\left[\underline{T}> t \right]&\eqdef \max\{I_1(t)-I_2(t), 0
   \},
   \nonumber
\end{align}
implying
\begin{equation}
T_m^{(i)} \geqd \underline{T}. \nonumber
\end{equation}

Next, by Churnoff bound, we obtain $I_2(t)\leq O\left( e^{-\eta
t}\right)$ for some $\eta>0$. Using (\ref{eq:upboundNm}), we
 derive
$$
\lim_{t\rightarrow \infty}\frac{\log \Pr\left[\underline{N}>
  \frac{2t}{\expect[\underline{X}]}+1\right]}{\log t} = -\frac{\mu}{(M-1)(\lambda \wedge
  \nu)},
$$
 which,
 combined with the estimates for $I_2(t)$, implies that
\begin{equation}\label{eq:preTmLow4}
\lim_{t\rightarrow \infty}\frac{\log \Pr\left[\underline{T}
>t \right]}{\log t} = -\frac{\mu}{(M-1)(\lambda \wedge \nu)}.
\end{equation}

Combining (\ref{eq:preTmUp2}) and (\ref{eq:preTmLow4}) completes the
proof.
\end{proof}

The following lemma studies the distribution of the number of
retransmissions that occur from a point when there is a departure
until the system becomes full.
%To this end, we introduce the
%sequence $\{C_n\}_{n \ge 1}$ representing the time points when
%either a collision or a departure occurs with $C_0=0$.
For the two sequences $\{C_i\}$ and $\{D_m\}$ defined in
Subsection~\ref{sec:description},  noting that $\{D_m\}$ is a
subsequence of $\{C_i\}$,  we can define the position of $D_m$ in
$\{C_i\}$ by $h_m\eqdef \min\{i\ge 0: C_i = D_{m}\}$.
 Let $N^f_m$, $m\ge 0$ be the total number of
both collisions and departures until the system becomes full and all
the users are backlogged (a collision occurs) for the first time
after $D_{m}$, i.e., $N^f_m\eqdef \min\{l-h_m: U(C_l+)=M, l \ge
h_m\}$, where $U(C_l+)$ represents the right hand limit of $U(t)$ at
time $C_l$. Recall that $\left\{L_i^{(t)} \right\}_{1\le i\le U(t)}$
represents the packet sizes of all the $U(t)$ number of active users
at time $t$.
\begin{lemma}\label{lemma:Nfm}
For any finite values  $\{L_i^{(D_{m})}\}_{1\le i\le U(D_{m})}$ at
time $D_{m}$, uniformly for all $m> 0$, we have
\begin{equation}\label{eq:NsBound}
\Pr\left[N^f_m>n\right]=O\left(e^{-\eta \sqrt{n}}\right)
\end{equation}
where the constant $\eta>0$ does not depend on
$\{L_i^{(D_m)}\}_{1\le i\le U(D_m)}$.
\end{lemma}
\begin{remark}
 We believe
that it is possible to prove a tighter exponential bound
$\Pr\left[N^f_m>n\right]=O\left(e^{-\eta n}\right)$, but the
preceding Weibull bound suffices for our proofs.
\end{remark}

The \textbf{proof} of Lemma~\ref{lemma:Nfm} is presented in
Section~\ref{sc:ALOHAproof}.

\section{Stability}\label{section:stability}
In this Section, we derive the stability condition of finite
population ALOHA with variable packets. Corollaries~\ref{cor:pre1}
and \ref{cor:pre2} are based on
Proposition~\ref{theorem:preliminary};
Proposition~\ref{theorem:steadyUpper} studies the distributional
properties of the upper bound for the number of (re)transmissions
and transmission delay for each successfully received packet
observed at the receiver.  Using these results,  we derive the
stability condition in Theorem~\ref{theorem:refinedStability}.

We use $\underline{\varlimsup}$ to denote both $\varlimsup$ and
$\varliminf$, i.e., $\underline{\varlimsup}$ means that the
corresponding two statements with respect to $\varlimsup$ and
$\varliminf$ are true. From Proposition~\ref{theorem:preliminary},
we can easily obtain the following two corollaries. Note that in
Corollary~\ref{cor:pre1} we use $\underline{\varlimsup}$ with
respect to $m$ since the existence of the stationary region for
$N_m^{(i)}$ and $T_m^{(i)}$ is not established. At this point of our
analysis, we could not find an easy argument for resolving this,
maybe minor, technical issue.
\begin{corollary}\label{cor:pre1}
If $\lambda=\nu>0$, then, as $n\to \infty$,
 \begin{equation}
  \underset{m \rightarrow \infty}{\underline{\varlimsup}} \frac{\log \Pr\left[N_m^{(i)}>n\right]}{\log
  n}\to
  -\frac{\mu}{(M-1)\nu}. \nonumber
   \end{equation}
\end{corollary}
\begin{corollary}\label{cor:pre2}
If $0<\lambda\leq \nu$ and $\mu >  (M-1) \nu$, then the system has a
positive throughput. If $\lambda\geq \nu>0$ and $\mu<(M-1) \nu$,
then the system has a zero throughput.
\end{corollary}
\begin{proof}
Let $N(t)\eqdef \min\{j: \sum_{m=1}^{j}T_m \leq t \}$ be the
counting process for the number of successfully transmitted packets
observed at the receiver from time $0$ until time $t$. By the same
fashion, we can define the counting process $N^{(i)}(t)\eqdef
\min\{j: \sum_{m=1}^{j}T_m^{(i)} \leq t \}$ for user $i, 1\leq i\leq
M$, which represents the number of successfully transmitted packets
observed at user $i$ from time $0$ until time $t$. Clearly,  we have
\begin{equation}\label{eq:throughput}
  N(t) = \sum_{i=1}^{M}N^{(i)}(t)
\end{equation}
where $N(t), N^{(i)}(t)$ all go to infinity almost surely as $t\to
\infty$.

  Recalling the proof corresponding to $\bar{T}$ and
  $\underline{T}$ in Proposition~\ref{theorem:preliminary}, we can always
construct on the same probability space
$\left\{\bar{T}_j\right\}_{j\ge 1}$ and
$\left\{\underline{T}_j\right\}_{j\ge 1}$,  two sequences of
  i.i.d. copies of $\{\bar{T}\}$ and
  $\{\underline{T}\}$,  such that $\underline{T}_j
\leq T_m^{(i)} \leq \bar{T}_j$. Define
 $\overline{N}^{(i)}(t)\eqdef \min\{j: \sum_{m=1}^{j}\underline{T}_m^{(i)}
\leq t \}$ and
 $\underline{N}^{(i)}(t)\eqdef \min\{j: \sum_{m=1}^{j}\overline{T}_m^{(i)}
\leq t \}$ for user $i, 1\leq i\leq M$. By the preceding
definitions,  we can easily obtain
\begin{equation}  \label{eq:throuputSanwitch}
\underline{N}^{(i)}(t) \leq  N^{(i)}(t) \leq  \overline{N}^{(i)}(t).
\end{equation}

% Therefore,
%$$
%\sum_{j=1}^{N^{(i)}(t)} \underline{T}_j  \leq t \leq
%   \sum_{j=1}^{N^{(i)}(t)+1} \overline{T}_j.
%$$
%which, for $N^{(i)}(t)>1$, implies
%\begin{align}
%     \frac{N^{(i)}(t)}{\sum_{j=1}^{N^{(i)}(t)+1} \bar{T}_j }  \leq \frac{N^{(i)}(t)}{t}
%     \leq
%      \frac{N^{(i)}(t)}{\sum_{j=1}^{N^{(i)}(t)} \underline{T}_j
%      }. \nonumber
%  \end{align}

%Thus, if $\lambda\leq \nu$ and $\mu >  (M-1) \nu$, then
%\begin{align}\label{eq:throughput1}
%  \varliminf_{t \to \infty} \frac{N^{(i)}(t)}{t} &\geq
%  \varliminf_{t \to \infty} \frac{N^{(i)}(t)}{\sum_{j=1}^{N^{(i)}(t)+1} \bar{T}_j
%  } = \frac{1}{\expect\left[\bar{T}\right]} > 0,
%  \end{align}
Thus, if $\lambda\leq \nu$ and $\mu >  (M-1) \nu$, then
\begin{align}\label{eq:throughput1}
  \varliminf_{t \to \infty} \frac{N^{(i)}(t)}{t} &\geq
  \varliminf_{t \to \infty} \frac{\underline{N}^{(i)}(t)}{t} = \frac{1}{\expect\left[\bar{T}\right]} > 0,
  \end{align}
  since $\bar{T}$ has a power law tail with index greater than one ($\expect\left[\bar{T}\right]< \infty$)
  by Proposition \ref{theorem:preliminary}.

If $\lambda\geq \nu$ and $\mu<(M-1) \nu$, then
\begin{align}\label{eq:throughput2}
  \varlimsup_{t \to \infty} \frac{N^{(i)}(t)}{t} &\leq
  \varlimsup_{t \to \infty} \frac{\overline{N}^{(i)}(t)}{t} = \frac{1}{\expect\left[\underline{T}\right]} = 0,
  \end{align}
since  $\underline{T}$ has a power law tail with index smaller than
one ($\expect\left[\bar{T}\right]= \infty$) by Proposition
\ref{theorem:preliminary}.

 Combining (\ref{eq:throughput}), (\ref{eq:throughput1}) and
 (\ref{eq:throughput2}), we finish the proof.
\end{proof}

\begin{proposition}\label{theorem:steadyUpper}
For an ALOHA system with finite size packets at $t=0$ and under
condition (\ref{eq:condition}) on asymptotically exponential packet
sizes, there exist $\hat{N}$ and $\hat{T}$ such that the number of
transmissions $N_m$ and the transmission time $T_m$ satisfy
$$N_m \leqd \hat{N},  \quad T_m \leqd
\hat{T}$$ with
\begin{equation}\label{eq:propUpperN}
  \lim_{n \rightarrow \infty} \frac{\Pr[\hat{N} >n]}{\log n}  =
  \lim_{t \rightarrow \infty} \frac{\Pr[\hat{T} >t]}{\log t} =- \frac{\mu}{(M-1)
  \nu}.
\end{equation}
%\begin{equation}\label{eq:propUpperT}
% \lim_{t \rightarrow \infty} \frac{\Pr[\bar{T} >t]}{\log t} = - \frac{\mu}{(M-1) \nu}.
%\end{equation}
\end{proposition}

\begin{proof}
Recalling the definition of $N_m^f$ before Lemma~\ref{lemma:Nfm} and
using the union bound, we obtain
 \begin{align}\label{eq:p1p2}
\Pr[N_m>n]&= \Pr\left[N_m>n, N^f_{m-1}< N_m \right]
\nonumber\\
&\quad \quad +\Pr\left[N_m>n, N^f_{m-1}\geq
N_m \right] \nonumber\\
&\leq \Pr\left[N_m- N^f_{m-1}+ N^f_{m-1}> n,
N^f_{m-1}<  N_m \right]\nonumber\\
& \quad \quad + \Pr\left[N^f_{m-1}>n \right] \nonumber\\
 &\leq \Pr\left[N_m- N^f_{m-1}> \frac{n}{2},  N^f_{m-1}<
 N_m \right] \nonumber\\
 &\quad \quad +\Pr\left[ N^f_{m-1}> \frac{n}{2} \right]+  \Pr\left[N^f_{m-1}> \frac{n}{2} \right].
 \end{align}
%Because of (\ref{eq:lowboundNm}),   the probability tail of $N_m$
%decays slower than the one of $N^f_m$, implying $\Pr\left[ N^f_m<
% N_m \right]>0$, and therefore,
%  \begin{align}\label{eq:p1p2}
% \Pr[N_m>n]&\leq \Pr\left[N_m- N^f_m> \frac{n}{2} \; {\big |}\;  N^f_m<
% N_m \right]+2\Pr\left[ N^f_m> \frac{n}{2}\right].
% \end{align}

By Lemma (\ref{lemma:Nfm}), we know that for some $0<\zeta<1$,
\begin{align}\label{eq:N1N2ab}
\Pr\left[ N^f_{m-1}> \frac{n}{2} \right]+  \Pr\left[N^f_{m-1}>
\frac{n}{2} \right] &\leq 2\Pr\left[ N^f_{m-1}> \frac{n}{2} \right]
\nonumber\\
& \leq 2\zeta^{\sqrt{\frac{n}{2}}}.
\end{align}

%By following the same arguments as for (\ref{eq:N1N2}) in the proof
%of the upper bound of Theorem~\ref{theorem:transient}, we obtain,
%for $0<\zeta<1$,
% \begin{align}\label{eq:p1p2}
%\Pr[N_m>2n]
% &\leq \Pr[N_m- N^f_m> n \; {\big |}\;  N^f_m<
% N_m]+ \zeta^{\sqrt{n}}.
% \end{align}

%Observe that right after time $D_{m-1}$, one user has just
%successfully transmitted a packet, and thus there are at most $M-1$
%active users right after time $D_{m-1}$.
Observe that $N^f_{m-1}<
 N_m$ implies that there exists time $\sigma<D_m$ at which each user has a packet and
 is in backoffed status. Thus, by recalling the notation defined for Lemma~\ref{lemma:Nfm}, we  can denote the
packet sizes held by $M$ active users at time $\sigma$  by
$L_{i}^{(\sigma)}, 1\le i \le M$.  In addition, we know that right
after time $D_{m-1}$, one user has just successfully transmitted a
packet. Thus, at time $\sigma$ (when the system is full) there is at
least one new packet with size equal in distribution to $L$ in the
system. Therefore,
% and recalling that
%$\left\{L_i^{(D_{m}-)}\right\}$ denote the packet sizes observed
%right before time $D_{m}$,
 we obtain
\begin{align*}
  &\Pr\left[N_m- N^f_{m-1}> \frac{n}{2},  N^f_{m-1}<
 N_m \right]\nonumber\\
 & \quad \quad \leq \expect\left[ \left(
1-\frac{1}{M}\left(  \sum_{i=1}^{M}e^{-L_i^{(D_{m}-)}(M-1)\nu}
            \right) \right)^{\lfloor\frac{n}{2}\rfloor} \right]\\
            & \quad \quad \leq \expect\left[ \left(
1-\frac{1}{M} e^{-L(M-1)\nu}
             \right)^{\lfloor\frac{n}{2}\rfloor} \right],
\end{align*}
which, in conjunction with (\ref{eq:p1p2}) and (\ref{eq:N1N2ab}),
implies, uniformly for all $m$,
 \begin{align*}
\Pr[N_m>n]
 &\leq \expect\left[ \left(
1-\frac{1}{M} e^{-L(M-1)\nu}
             \right)^{\lfloor\frac{n}{2}\rfloor} \right]+ 2\zeta^{\sqrt{\frac{n}{2}}}.
 \end{align*}

Now, we can define a  random variable $\hat{N}$ which satisfies, for
integer $n$, $\Pr[\hat{N}>n]$ is equal to
$$
 \min \left\{ 1,  \expect\left[ \left(
1-\frac{1}{M} e^{-L(M-1)\nu}
             \right)^{\lfloor\frac{n}{2}\rfloor} \right]+ 2\zeta^{\sqrt{\frac{n}{2}}}
             \right\},
$$
implying
$$
N_m \leqd \hat{N}.
$$

 By using the same approach as in calculating (\ref{eq:preliminaryUp}), we obtain
\begin{equation*}
\lim_{n\rightarrow \infty}\frac{\log \Pr[\hat{N}>n ]}{\log n} =
-\frac{\mu}{(M-1)\nu},
\end{equation*}
which finishes the proof of the result on $\hat{N}$ in equation
(\ref{eq:propUpperN}). The proof for $\hat{T}$ follows  similar
arguments as in proving the result on $T_m^{(i)}$ in Proposition
\ref{ss:pre}.
\end{proof}

Combining Theorem \ref{theorem:steadyUpper} and Corollary
\ref{cor:pre2}, we obtain the following theorem. Observe that this
theorem is slightly more general than Corollary~\ref{cor:pre2} since
it shows that $\mu>(M-1)\nu$ is enough for positive throughput,
i.e., the additional condition $\lambda\leq \nu$ in
Corollary~\ref{cor:pre2} is not needed.
\begin{theorem}\label{theorem:refinedStability}
Under condition (\ref{eq:condition}),  if $\mu>(M-1)\nu$, the ALOHA
system has a positive throughput.  Conversely, if $\lambda\geq
\nu>0$ and
   $\mu<(M-1)\nu$, then, the system has a zero throughput.
\end{theorem}
\begin{remark}
For the critical case $\mu = (M-1)\nu$, if $L$ has an exact
exponential tail, i.e.,  $\Pr[L>x] \sim ce^{-\mu x}$ and $\lambda\ge
\nu$, then, the limiting distributions of $N_m^{(i)}$ and
$T_m^{(i)}$ would have exact power law tails of index $1$, and
therefore, have infinite means.
\end{remark}
\begin{remark}
%We conjecture that
%\begin{equation*}
%\lim_{\lambda \to 0}\lim_{n\rightarrow \infty}\frac{\log
%\Pr\left[N_m^{(i)}>n \right]}{\log n} = -\frac{M\mu}{(M-1)\nu}.
%\end{equation*}
The condition $\lambda\ge \mu$ and
   $\mu < (M-1)\nu$ yields a zero
   throughput. However, it appears that one
could obtain a positive throughput by decreasing $\lambda$ for fixed
$\nu$ and $\mu$ in this case. Specifically, we conjecture that the
throughput of the system is positive
 when  $\lambda$ is small enough and $M\mu>(M-1)\nu>\mu$.
\end{remark}

\begin{proof}
The second statement of this theorem is the same as the second
statement of Corollary~\ref{cor:pre2}. Given
Proposition~\ref{theorem:steadyUpper}, the first statement can be
easily derived using basically the same arguments as in the proof of
Corollary~\ref{cor:pre2}, and thus we omit the details.
\end{proof}

\section{Approximation of the Distributions of $N_m$ and
$T_m$}\label{section:distributional}
\subsection{Starting Behavior}\label{ss:starting}

In this subsection, we study  the number of retransmissions $N_m$
and the transmission delay $T_m$ for the $m$th successfully
transmitted packet observed at the receiver when the system starts
from an empty state. This result characterizes the starting behavior
of our ALOHA model for small (finite) $m$. Furthermore, since ALOHA
tends to accumulate with time longer packets, it would make sense to
define a modified ALOHA which,  after a finite (possibly large)
number of successful transmissions, refreshes itself by discarding
all the packets currently present in the system.  Hence, for this
modified ALOHA, the following theorem describes the steady state
behavior as well.

\begin{theorem}\label{theorem:transient}
Under condition (\ref{eq:condition}), assume that at time $t=0$ the
system is empty $U(0)=0$,  then, for any fixed $m\geq M$,   the
number of transmissions $N_m$ and the transmission time $T_m$
 satisfy
\begin{equation}\label{eq:propLowerN}
 \lim_{n \rightarrow \infty} \frac{\Pr[N_m >n]}{\log n}
   =\lim_{t \rightarrow \infty} \frac{\Pr[T_m >t]}{\log t}= - \frac{M\mu}{(M-1)
   \nu}.
\end{equation}
%\begin{equation}\label{eq:propLowerT}
% \lim_{t \rightarrow \infty} \frac{\Pr[T_m >t]}{\log t}  = - \frac{M\mu}{(M-1)
% \nu}.
%\end{equation}
\end{theorem}
\begin{remark}
A special case of this theorem when $U(C_m+)=M$ with all the packets
in the system  being i.i.d. and equal in distribution to $L$ was
proved in Theorem 1 of \cite{Jelen06ALOHA}.
\end{remark}
\begin{remark}
Note that this result still holds even if we allow $m$ to be a
slowly growing function of $n$ for $N_m$, e.g., $m=o(\log n)$ (or
$m=o(\log t)$ for $T_m$).
\end{remark}
\begin{remark}
 This theorem indicates that the distribution tails
of $N_m$ and $T_m$ are essentially power laws when the packet
distribution is approximately exponential ($\approx e^{-\mu x}$).
Thus, the finite population ALOHA may exhibit high variations, e.g.,
 the system has infinite average transmission time when
$0<M\mu /(M-1)\nu<1$; and when $1<M\mu /(M-1)\nu<2$, the
transmission time has finite mean but infinite variance.
 It might be worth noting that this may even occur
when the expected packet length is much smaller than the expected
backoff time $\expect L \ll 1/\nu$.
\end{remark}
\begin{proof}[Proof of Theorem \ref{theorem:transient}]
We first prove the logarithmic asymptotics for $N_m$, based on which
a similar result can be proved for $T_m$.

First, we begin with proving the \emph{lower bound} for $N_m$.  We
construct a special event with a positive probability that guides
the system from time $0$ up to time $D_{m-1}$. Denote by
$\mathscr{E}_1$ the event that only one of the users has packets to
send  and all the other $M-1$ users are empty from time $0$ through
time $D_{m-1}$; additionally, we require that the sizes of these
arriving packets be less than a constant $k-1$ with $\Pr[L\le
k-1]>0$ and that each new arrival be within a unit interval after
the previous departure. This construction implies $D_{m-1}\le
(m-1)k$, and therefore, by time $D_{m-1}$ the probability that the
system evolves according to $\mathscr{E}_1$ is lower bounded by
\begin{equation}
 \Pr[\mathscr{E}_1] \geq  \left( (1-e^{-\lambda}) \Pr[L\leq k-1] \right)^{m-1} e^{-(M-1)\lambda
 (m-1)
  k}>0.
\end{equation}

Next, immediately after time $D_{m-1}$, observe that the whole
system becomes empty according to our construction. Then, we build
another special event $\mathscr{E}_2$ that leads $M$ users to have
i.i.d. packets with sizes that are larger than $1$ in their buffers
after time $D_{m-1}$.

To this point, we require that each of the $M$ users have a packet
with size larger than $1$ arriving to the system after $D_{m-1}$ and
that their arriving points be within $[D_{m-1}, D_{m-1}+1]$. This
event happens with probability $(1-e^{-\lambda})^M \Pr[L>1]^M $.
Notice that, immediately after the $M$th packet arrives,  there are
either $M-1$ or $M$ users in the backoff status, depending on
whether the $M$th arrived packet collides with others upon arrival
or not. If the $M$th packet does not collide with others upon
arrival, we require that a retransmission occur within one unit of
time after it arrives, which happens with a probability greater than
$1-e^{-(M-1)\nu}$. These requirements can guarantee that there
exists a time $\tau \in [D_{m-1}, D_{m-1}+2)$ with $ \tau = \min
\{C_n \mid U(C_n+)=M, C_n>D_{m-1}\}$,  at which each user in the
system has a packet and is in the backoff status.  The probability
that the event $\mathscr{E}_2$ happens is lower bounded by
$$
  \Pr[\mathscr{E}_2] \geq (1-e^{-\lambda})^M  \Pr[L>1]^M (1-e^{-(M-1)\nu})>0.
$$

Now, given $\mathscr{E}_1$ and $\mathscr{E}_2$, we can denote by
$N_{m}^{\ast}$ the number of retransmissions between $(\tau, D_m]$,
implying $N_m
 \geq N_{m}^{\ast}$.  Then, recalling the notation defined before Lemma~\ref{lemma:Nfm} and defining
$L_{o} \eqdef \min
\left\{L_1^{(\tau)},L_2^{(\tau)},\cdots,L_M^{(\tau)} \right\}$,  we
obtain
\begin{align}\label{eq:I4}
  \Pr[N_{m}^{\ast} >n \mid \mathscr{E}_1, \mathscr{E}_2] &=\expect\left[ \left(
1-\frac{1}{M}\left(  \sum_{i=1}^{M}e^{-L_i^{(\tau)}(M-1)\nu}
            \right) \right)^{n} \right] \nonumber\\
  & \geq \expect \left[\left(1-e^{-L_o (M-1)\nu} \right)^n \right].
  \end{align}
It is  easy to check that the complementary cumulative distribution
function $\Pr[L_o \geq x]$ satisfies
\begin{equation*}
   \lim_{x\rightarrow \infty}\frac{\log \Pr[L_o \geq x]}{x}=-M\mu,
   \end{equation*}
which,  by using the same technique as in estimating
(\ref{eq:templateLowbound5}),  yields
\begin{equation}%\label{eq:lowerN}
 \varliminf_{n \rightarrow \infty}\frac{\log  \Pr[N_{m}^{\ast} >n \mid \mathscr{E}_1, \mathscr{E}_2]}{\log n} \geq
-\frac{M\mu}{(M-1)\nu}. \nonumber
\end{equation}
Finally,  using $\Pr[N_m>n] \geq \Pr[\mathscr{E}_1,
\mathscr{E}_2]]\Pr[N_{m}^{\ast} >n \mid \mathscr{E}_1,
\mathscr{E}_2]$ completes the proof of the lower bound for $N_m$.

Next, we proceed with the proof of the \emph{upper bound} for $N_m$.
%By recalling the definition of $N^f_m$ in Lemma (\ref{lemma:Nfm})
%and using the union bound, we obtain,
% \begin{align}\label{eq:N1N2}
%\Pr[N_m>n]&= \Pr\left[N_m>n, N^f_m< N_m \right]+\Pr\left[N_m>n,
%N^f_m\geq
%N_m \right] \nonumber\\
%&\leq \Pr\left[N_m- N^f_m+ N^f_m> n,
%N^f_m<  N_m \right]  + \Pr\left[N^f_m>n \right] \nonumber\\
% &\leq \Pr\left[N_m- N^f_m> \frac{n}{2},  N^f_m<
% N_m \right]+\Pr\left[ N^f_m> \frac{n}{2} \right]+  2\Pr\left[N^f_m> \frac{n}{2} \right].
% \end{align}
Using the same approach as in evaluating (\ref{eq:p1p2}), we obtain
 \begin{align}\label{eq:N1N2}
\Pr[N_m>n]&\leq \Pr\left[N_m- N^f_{m-1}> \frac{n}{2},  N^f_{m-1}<
 N_m \right]+
2\zeta^{\sqrt{\frac{n}{2}}}.
 \end{align}

%
%By Lemma (\ref{lemma:Nfm}), we know that for some $0<\zeta<1$,
%\begin{equation}\label{eq:N1N2a}
%2\Pr\left[ N^f_m> \frac{n}{2} \right] \leq
%2\zeta^{\sqrt{\frac{n}{2}}}.
%\end{equation}

Now,  we observe that $N^f_m<
 N_m$ implies that there exists time $\sigma<D_m$ at which each user has a packet at hand and
 is in the backoff status. Thus,  we denote the
packet sizes held by $M$ users at time $\sigma$  by
$L_{i}^{(\sigma)}, 1\le i \le M$.  In addition, we know that at time
$\sigma$ the total number of packets, including those still present
in the system and those already successfully transmitted, is less
than $m+M$ since the system has only $M$ users. Denote the sizes of
the first $m+M$ packets arriving to the system by $\{L_1, L_2,
\cdots, L_{m+M}\}$ and its order statistics by $L^{(1)} \geq L^{(2)}
\geq \cdots \geq L^{(m+M)}$, and we obtain
\begin{align}\label{eq:hi0}
  \Pr[N_m- N^f_{m-1} & > n ,  N^f_{m-1} <
 N_m] \nonumber\\
 &\leq \expect\left[ \left(
1-\frac{1}{M}\left(  \sum_{i=1}^{M}e^{-L_i^{(\sigma)}(M-1)\nu}
            \right) \right)^{n} \right] \nonumber\\
            &\leq \expect\left[ \left(
1-  \frac{1}{M}e^{-L^{(M)}(M-1)\nu}
             \right)^{n} \right].
           %   &=\expect\left[ \left(
%1- \frac{1}{M}e^{-L^{(M)}(M-1)\nu} \right)^{n} \ind \left(L^{(M)}>
%x_{\epsilon}\right)\right]\nonumber  \\
%  &\;\;\;\;\;\;\;  +\expect\left[ \left(
%1- \frac{1}{M}e^{-L^{(M)}(M-1)\nu} \right)^{n} \ind \left(L^{(M)}
%\leq x_{\epsilon}\right)\right] \nonumber.
\end{align}
 Since
$L^{(M)}$ is the $M$th largest value among $L_i, 1\leq i \leq m+M$,
we know
$$
     \lim_{x \to \infty} \frac{\Pr[L^{(M)}>x]}{x} = -M \mu.
$$
Then, by (\ref{eq:N1N2}), (\ref{eq:hi0}) and using the same approach
as in estimating (\ref{eq:preliminaryUp}), one derives
\begin{align}\label{eq:upperN}
\varlimsup_{n\rightarrow \infty}\frac{\log \Pr[N_m>n ]}{\log n} \leq
-\frac{M\mu}{(M-1)\nu},
\end{align}
%. By the stochastic
%ordering described in (\ref{eq:1}) of Lemma \ref{lemma:1}, for a
%uniform random variable $0\leq U\leq 1$ and $x_{\epsilon}$ large
%enough,  we have,
%\begin{align*}%\label{eq:hi2}
%\Pr[N>n ]  &\leq
%%\left( \expect
%%\left[e^{-\frac{n}{M}U^{(1+\epsilon)(M-1)\nu/(\mu-\epsilon)}}
%%     \right] \right)^M\nonumber
%%   \\
%%&\leq \left( \expect
%%\left[e^{-\frac{n}{M}U^{(M-1)\nu/(\mu-\epsilon)}}
%%     \right] \right)^M\nonumber
%%   \\
% \left( \frac{M^{(\mu-\epsilon)/(1-\epsilon)(M-1)\nu} \Gamma \left((\mu-\epsilon)/(1-\epsilon)(M-1)\nu+1 \right)}
%      {n^{(\mu-\epsilon)/(1-\epsilon)(M-1)\nu}} \right)^M + \eta^{n} ,
%\end{align*}
which completes the proof of the upper bound.

The proof for the logarithmic asymptotics of $T_m$ is based on
similar arguments as in proving $T_m^{(i)}$ in
Proposition~\ref{theorem:preliminary} and, thus, we omit the
details.
%Now,  by writing $T_{m+1} = \tau + T$ and noting $\Pr[N_m
%>n]\geq \Pr[N_m >n, \mathscr{E}_1, \mathscr{E}_2]$,  it is easy to prove that, using
%Lemma \ref{lemma:exp},
%\begin{equation}
% \varliminf_{n \rightarrow \infty} \frac{\Pr[N_{m}^{\ast} >n \mid \mathscr{E}_1, \mathscr{E}_2]}{\log n}  \geq - \frac{M\mu}{(M-1)
% \nu}.
%\end{equation}

\end{proof}

\subsection{Limiting Steady State Behavior}\label{ss:steadyState}

When the system keeps running for a long period of time,  we can
show that the preceding upper bound,  presented in
Theorem~\ref{theorem:steadyUpper}, is attainable when
$\lambda=\nu>\mu/(M-1)$.   In order to study this situation, first
we establish the following lemma that characterizes the growth of
the packet sizes in the system immediately after a departure at time
$D_m$.  Noting that $\lambda=\nu$, we can assume that once a user
successfully transmits a packet through the channel, it immediately
generates a new packet in its buffer and goes into the backoff
state, i.e., we can interpret that the arrival and departure happen
at the same time. Therefore, the system evolves as if it always had
$M$ packets available and all users remained in the backoff state
over the entire operation. Let $\underline{L}(D_m)$ be the minimum
of the packet sizes of the other $M-1$ users except the one
departing at time $D_m$.

\begin{lemma}\label{lemma:heavyAssmp}
Assume that $\lambda=\nu>\mu/(M-1)$ and
\begin{equation}\label{eq:steadyCondition}
\varlimsup_{y\to \infty}\sup_{\delta y < x < y} \frac{1}{y-x} \log
\left( \frac{\Pr[L>x]}{\Pr[L>y]} \right) \leq \mu
\end{equation}
for $0<\delta<1$.  Then, there exists $p>0$ such that for any fixed
$y$,
\begin{equation}
\varliminf_{m \to \infty}\Pr\left[\underline{L}(D_m) > y \right]
>p.
\end{equation}
\end{lemma}
\begin{remark}
We believe that a stronger result  $\varliminf_{m \to
\infty}\Pr\left[\underline{L}(D_m) > y \right] =1$ for all $y$ is
also true, but the preceding lemma suffices for our proofs.
Furthermore, a careful examination of our proof shows that the
result is also true for $\min \{\lambda,\nu \}>\mu/(M-1)$, but we
avoid this generalization due to considerable notational
complications.
\end{remark}
\begin{remark}
It is easy to see that condition (\ref{eq:steadyCondition}) holds
for a broad range of distributions from exponential family, e.g.,
Gamma distribution,  $e^{-\mu x}e^{\gamma x^{\beta}}$ with
$0<\beta<1$, etc.
\end{remark}

The \textbf{proof} of Lemma~\ref{lemma:heavyAssmp} is presented in
Section~\ref{sc:ALOHAproof}.  By using this lemma,  we can derive
the following theorem that characterizes the limiting steady state
behavior of our ALOHA model.

\begin{theorem}\label{theorem:steadyLower}
Under condition (\ref{eq:steadyCondition}), if
$\lambda=\nu>\mu/(M-1)$, we obtain
\begin{align}\label{eq:steadyLower}
 \lim_{n \to \infty} \underset{m \rightarrow \infty}{\underline{\varlimsup}} \frac{\Pr[N_m >n]}{\log n} &  =
 \lim_{t \to \infty} \underset{m \rightarrow \infty}{\underline{\varlimsup}} \frac{\Pr[T_m >t]}{\log t} \nonumber\\
 & =- \frac{\mu}{(M-1)
  \nu}.
\end{align}
\end{theorem}
%\begin{remark}
%If $\lambda=\nu<\mu/(M-1)$, we conjecture
%\begin{equation}\label{eq:conjecture}
% \underset{n \to \infty}{\lim} \underset{m \rightarrow \infty}{\underline{\varlimsup}} \frac{\Pr[N_m >n]}{\log n}  =
% \lim_{t \to \infty} \underset{m \rightarrow \infty}{\underline{\varlimsup}} \frac{\Pr[T_m >t]}{\log t} =- \frac{M\mu}{(M-1)
%  \nu}+M-1.
%\end{equation}
%We further verify this conjecture in Section \ref{s:NSE} on
%simulation experiments.
%\end{remark}
\begin{proof}
First, we prove the result for $N_m$.  The upper bound is implied by
Proposition~\ref{theorem:steadyUpper} and thus, we only need to
prove the lower bound. Recalling the definition of
$\underline{L}(D_m)$ in the paragraph before
Lemma~\ref{lemma:heavyAssmp} and using Lemma \ref{lemma:heavyAssmp},
we obtain that there exist $p>0$ and $m_0>0$ such that  for all
$m>m_0$,
\begin{equation}\label{eq:minL}
\Pr\left[ \underline{L}(D_{m-1})
> \frac{\log n}{(M-1)\nu} \right]>p.
\end{equation}

Since there is a new packet with size equal in distribution to $L$
arriving to the system at time $D_{m-1}$ (see the discussion before
Lemma~\ref{lemma:heavyAssmp}), and the packet sizes of the other
$M-1$ users are lower bounded by $\underline{L}(D_{m-1})$, we obtain
\begin{align}\label{eq:steadPn}
  \Pr[N_m >n ] &\geq \Pr\left[N_m >n,  \underline{L}(D_{m-1}) > \frac{\log n}{(M-1)\nu} \right]\nonumber\\
   &\geq \expect{\Bigg[} \left(
1-\frac{1}{M}\left(\sum_{i=1}^{M}e^{-L_i^{\left(D_{m-1}\right)}(M-1)\nu}
            \right) \right)^{n} \nonumber\\
            &\quad \quad  \times \ind\left( \underline{L}(D_{m-1})
> \frac{\log n}{(M-1)\nu} \right) {\Bigg]} \nonumber\\
  &\geq  \expect{\Bigg[} {\Bigg(}
1-\frac{1}{M}{\Bigg(}e^{-L(M-1)\nu} \nonumber\\
&\quad \quad \quad \quad \quad \quad \quad \; +
\sum_{i=1}^{M-1}e^{-\underline{L}(D_{m-1})\cdot(M-1)\nu}
            {\Bigg)} {\Bigg)}^{n} \nonumber\\
            &\quad \quad \times \ind\left( \underline{L}(D_{m-1})
> \frac{\log n}{(M-1)\nu} \right) {\Bigg]} \nonumber\\
  & \geq \Pr\left[\underline{L}(D_{m-1})
> \frac{\log n}{(M-1)\nu} \right] \nonumber\\
&\quad \quad \times \expect \left[\left(1-\frac{M-1}{M}\cdot
\frac{1}{n}-\frac{1}{M}e^{-L (M-1)\nu} \right)^n \right],
  \end{align}
  where we use the independence between the new packet size and $\underline{L}(D_{m-1})$ at time $D_{m-1}$
   in the last inequality.

Combining (\ref{eq:minL}) and (\ref{eq:steadPn})  yields, for $n$
large enough, $\Pr[N_m >n ] $ is lower bounded by
\begin{align}\label{eq:steadPn2}
  & p \left(1-\frac{M-1}{M}\cdot
  \frac{1}{n}\right)^n \nonumber\\
  &\quad \quad \times
     \expect\left[ \left( 1-\frac{1}{M(1-(M-1)/(Mn))}e^{-L(M-1)\nu} \right)^n
     \right] \nonumber\\
     &\quad \geq  p \left(1-\frac{M-1}{M}\cdot \frac{1}{n}\right)^n
     \expect\left[ \left( 1-e^{-L(M-1)\nu} \right)^n
     \right], \nonumber
  \end{align}
  which, by noting that $$\lim_{n \to \infty}\left(1-\frac{M-1}{M}\cdot
  \frac{1}{n}\right)^n = e^{-(M-1)/M}>0$$ and using the same approach as in calculating
  (\ref{eq:templateLowbound5}),
 completes the proof the lower bound.  The result on $T_m$ can be proved
  by using the same approach as in proving the result on $T_m$ in Proposition~\ref{theorem:preliminary}.
\end{proof}

\section{Power Laws in Slotted ALOHA with Random Number of Users}\label{s:slotted}
%\begin{figure}
%\centering \mbox{ \subfigure[Random number of objects on the
%pallet]{\epsfig{figure=pallets.eps,width=5cm}}\quad
%\subfigure[Random number of active nodes in the
%neighborhood]{\epsfig{figure=sensor.eps,width=5cm}}}
%\caption{objects on the pallet }
%\end{figure}

%\begin{figure}
%\centering
%\includegraphics[width=6.5cm]{pallets.eps}
 % \caption{Random number of objects on the pallet}
%\label{fig:pallets}
%\end{figure}

It is clear from the preceding section that the power law delays
arise due to the combination of collisions and packet variability.
Hence, it is reasonable to expect an improved performance when this
variability is reduced. Indeed, it is easy to see that the delays
are geometrically bounded in a slotted ALOHA with constant size
packets and a finite number of users. However, in this section we
will show that, when the number of users sharing the channel has
asymptotically an exponential distribution, the slotted ALOHA
exhibits power law delays as well. Situations with random number of
users are essentially predominant in practice, e.g., in sensor
networks, the number of active sensors in a neighborhood is a random
variable since sensors may switch between sleep and active modes, as
shown in Figure~\ref{fig:sensor}; similarly in ad hoc wireless
networks the variability of users may arise due to mobility, new
users joining the network, etc.

 \begin{figure}[h]
\centering
\includegraphics[width=5.5cm]{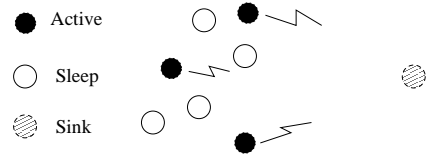}
  \caption{Random number of active neighbors in a sensor network.}
\label{fig:sensor}
\end{figure}
%\subsection{Logarithmic Asymptotics for Slotted ALOHA with Variable Number of Users}

More formally, consider a slotted ALOHA model (e.g., see
Section~4.2.2 of \cite{BG92}) with packets/slots of unit size and a
random number of users $M\geq 1$ that are fixed over time. This
model can be viewed as a first order approximation of a real system
where the number of users change very slowly.  Similarly as in
Section~\ref{s:finitePop}, each user holds at most one packet at a
time and after a successful transmission a new packet is generated
according to an independent Bernoulli process with success
probability $1-e^{-\lambda}, \lambda>0$. In case of a collision,
each colliding user backs off according to an independent geometric
random variable with parameter $e^{-\nu}, \nu>0$. Denote the number
of slots where transmissions are attempted but failed and the total
time between two successful packet transmissions as $N$ and $T$,
respectively.

\begin{theorem}\label{theorem:slot}
   If $\lambda = \nu$ and there exists $\alpha>0$, such that
\begin{equation*}
%\label{eq:Usercondition}
   \lim_{x\rightarrow \infty}\frac{\log \Pr[M>x]}{x}=-\alpha,
   \end{equation*}
    then, we have
   \begin{equation}\label{eq:NT}
   \lim_{n \rightarrow \infty}\frac{\log \Pr[N>n]}{\log
   n}=\lim_{t \rightarrow \infty}\frac{\log \Pr[T>t]}{\log
   t}=-\frac{\alpha}{\nu }.
   \end{equation}
  % and
%    \begin{equation}\label{eq:T}
%  \lim_{n \rightarrow \infty}\frac{\log \Pr[T>t]}{\log
%   t}=-\frac{\alpha}{\nu }.
%   \end{equation}
\end{theorem}
\begin{remark}
Similarly as in Theorem~\ref{theorem:transient}, this result shows
that the distributions of $N$ and $T$ are essentially power laws,
i.e., $\Pr[T>t]\approx t^{-\alpha/\nu}$ and, clearly, if
$\alpha<\nu$, then $\expect N=\expect T=\infty$.
\end{remark}

\begin{proof}
Since $\lambda = \nu$, we can consider a situation where all the
users are backlogged, i.e., have a packet to send. In this case the
total number of collisions between two successful transmissions is
geometrically
 distributed given $M$,
 $$ \Pr[N>n \mid M]=\left(1-\frac{Me^{-(M-1)\nu}(1-e^{-\nu})}{1-e^{-M\nu}} \right)^n,\; n\in
 \nat,
 $$
since, given $M$, $1-e^{-M\nu}$ is the conditional probability that
there is an attempt to transmit a packet, and
$1-e^{-M\nu}-Me^{-(M-1)\nu}(1-e^{-\nu})$ is the conditional
probability that there is a collision.  Therefore,
 \begin{equation}\label{eq:rep}
  \Pr[N>n]=\expect \left[\left(1-\frac{Me^{-(M-1)\nu}(1-e^{-\nu})}{1-e^{-M\nu}} \right)^n \right].
\end{equation}
On the other hand, we have
 \begin{equation}\label{eq:repT}
  \Pr[T>t]=\expect \left[\left(1-Me^{-(M-1)\nu}(1-e^{-\nu}) \right)^t \right], \; t\in \nat.
\end{equation}
 Now, following the same arguments as in the
proof of Proposition~\ref{theorem:preliminary}, we can prove
(\ref{eq:NT}).
%Similarly, one can show that the same asymptotic
%results hold if the initial number of backlogged users is  less than
%$M$. Due to space limitations,  a complete proof of this theorem
%will be presented in the extended version of this paper.

\end{proof}

  Actually, using part i) of Theorem~2.1 in \cite{Jelen07}, we can compute the
 exact asymptotics of $T$ under more restrictive conditions.
\begin{theorem}\label{theorem:exact}
%If  $\lambda=\nu$ and $\bar{F}(x)\eqdef \Pr[M>x]$ satisfies
%%$H_N\left(-\log \bar{F}(x)\right) \bar{F}(x)^{1/\beta} \thicksim
%%e^{-\nu x}$ and
% $H\left(-\log \bar{F}(x)\right)
%\bar{F}(x)^{1/\beta} \thicksim  x e^{-\nu x}$
%%, $\beta=\alpha/\nu$
%with $H(x)$ being continuous and regularly varying, then, as $t \to
%\infty$,
% \begin{equation*}\label{eq:asympN}
%%\Pr[N>n]\thicksim
%%\frac{\left(1-\expect\left[e^{-(M-1)\nu}\right]\right) \Gamma(
%%\beta+1 )}{ n^{\beta} H_N({\beta}\log n)^{\beta}},
%\Pr[T>t]\thicksim\frac{\Gamma( \beta+1 )(e^{\nu}-1)^\beta}{
%t^{\beta} H({\beta}\log t)^{\beta}}.
%\end{equation*}
If  $\lambda=\nu$ and $\bar{F}(x)\eqdef \Pr[M>x]$ satisfies
 $\bar{F}^{-1}(x)\thicksim  \Phi\left(e^{\nu x}(e^{\nu}x-x)^{-1}\right)$,
where $\Phi(\cdot)$ is regularly varying with index $\beta>0$, then,
as $t \to \infty$,
 \begin{equation*}\label{eq:asympN}
\Pr[T>t]\thicksim\frac{\Gamma( \beta+1 )}{\Phi(t)}.
\end{equation*}
\end{theorem}

\section{Simulation Examples}\label{s:NSE}
In this section, we illustrate our theoretical results with
simulation experiments. In particular, we emphasize the
characteristics of the studied ALOHA protocol that may not be
immediately apparent from our theorems. For example, in practice,
the distributions of packets and number of random users might have
bounded supports. We show that this situation may result in
truncated power law distributions for the transmission delays. To
this end, it is also important to note that the delay distribution
has a power law main body with a stretched support in relation to
the support of $L$ and $M$ and, thus, may result in very long,
although, exponentially bounded delays.

\begin{example}[Finite population model]\label{ex:1}
For the finite population model described in
Subsection~\ref{sec:description},  we compare the starting and
steady state behavior in this experiment.

\begin{figure}[h]
\centering
\includegraphics[width=8.6cm]{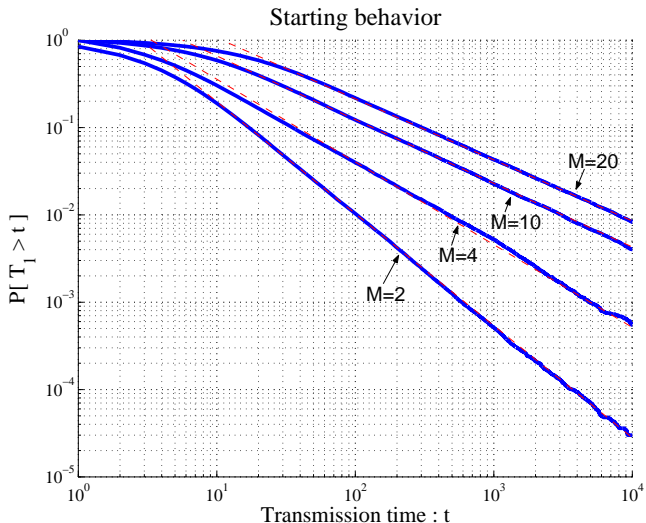}
\caption{Starting behavior: transmission time distribution for the
first successfully transmitted packet for finite population ALOHA
with variable size packets.} \label{fig:simulationR}
\end{figure}

 First, we verify Theorem~\ref{theorem:transient} on the starting behavior by plotting the empirical
distribution of time $T_1$ for the first successful transmission in
a system that is initially empty. In this regard, we conduct four
experiments for $M=2,4,10,20$ users, respectively. The packets are
assumed i.i.d. exponential with mean $1$ and the arrival intervals
and backoffs follow an exponential distribution with mean $2/3$. The
simulation experiments that each repeatedly measure $10^5$ samples
are shown in Figure~\ref{fig:simulationR}, which indicates a power
law transmission delay.  We can see from the figure that, as $M$
gets large ($M=10,20$), the slopes of the distributions that
represent the power law exponents on the $\log/\log$ plot are
essentially the same, as predicted by our
Theorem~\ref{theorem:transient}.

Next, we compare the starting behavior with the steady state
behavior predicted by Theorem~\ref{theorem:steadyLower}. In this
setting, we set $M=3$ and choose i.i.d. packet sizes that follow an
exponential distribution with mean $1$. In addition,  we assume that
arrival intervals and backoffs are exponential with mean $1.5$. The
starting behavior is represented by repeatedly measuring $10^5$
number of the transmission times for the first packet (m=1) in a
system that is initially empty and the steady state distribution is
obtained by continuously measuring the transmission times of the
packets with indexes from $m=10^5$ to $m=10^7$. The plot in Figure
\ref{fig:simulHeavy} shows that the transmission time distribution
of the first packet for the starting behavior has a slope
$-M\mu/((M-1)\nu)=-2.25$, and the steady state transmission time
distribution has a slope $-\mu/((M-1)\nu)=-0.75$, as predicted by
equations (\ref{eq:propLowerN}) and (\ref{eq:steadyLower})
 in the log-log scale, respectively.
\begin{figure}[h]
\centering
\includegraphics[width=8.6cm]{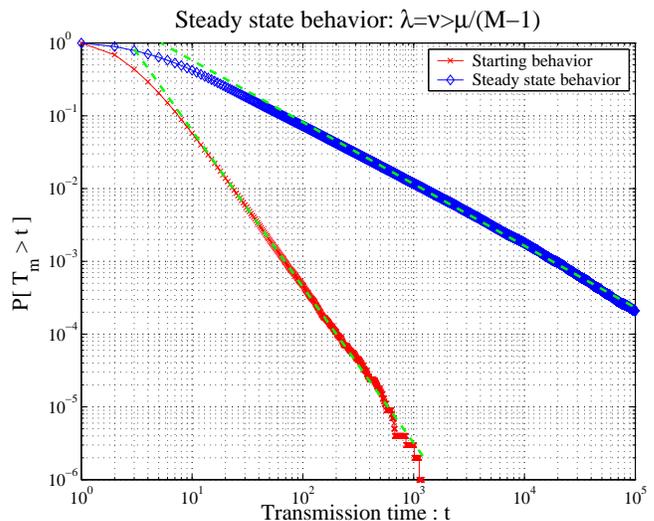}
\caption{Comparing starting behavior and steady state behavior for
finite population ALOHA with variable size packets.}
\label{fig:simulHeavy}
\end{figure}

%For the second experiment, arrival intervals and backoffs are
%assumed to follow exponential distribution with mean $2.4$ and all
%the other parameters are the same as the preceding experiment.  The
%plot in Figure \ref{fig:simulLight} shows that the transmission time
%distribution of the first packet for the starting behavior has a
%slope equal to $M\mu/((M-1)\nu)=3.6$, and that the slope of the
%steady state transmission time distribution is equal to
%$M\mu/((M-1)\nu)-M+1=1.6$, as predicted by equation
%(\ref{eq:conjecture}).
%\begin{figure}[h]
%\centering
%\includegraphics[width=8.5cm]{steadyStateMuGNu.eps}
%\caption{Steady state distribution of the interval between
%successfully transmitted packets for finite population ALOHA with
%variable size packets.} \label{fig:simulLight}
%\end{figure}
\end{example}

\begin{example}[Random number of users]\label{ex:2}
As stated in Section~\ref{s:slotted}, the situation when the number
of users $M$ is random may cause heavy-tailed transmission delays
even for slotted ALOHA.  However, in many practical applications the
number of active users $M$  may be bounded, i.e., the distribution
$\Pr[M>x]$ has a bounded support. Thus, from equation
(\ref{eq:repT}) it is easy to see that the distribution of $T$ is
exponentially bounded. However, this exponential behavior may happen
for very small probabilities, while the delays of interest can fall
inside the region of the distribution (main body) that behaves as
the power law. This example is aimed to illustrate this important
phenomenon. Assume that initially $M\geq 1$ users have unit size
packets ready to send and $M$ follows geometric distribution with
mean $3$. The backoff times of the colliding users and the arrival
intervals of the new packets are independent and geometrically
distributed with mean $2$.
 We take the number of users to have finite support $[1,K]$ and show how this results
in a truncated power law distribution for $T$ in the main body, even
though the tails are exponentially bounded. This example is
parameterized by $K$ where $K$ ranges from $6$ to $14$ and for each
$K$ we set the number of users to be equal to $M_K=\min(M,K)$.
 \begin{figure}
\centering
\includegraphics[width=8.6cm]{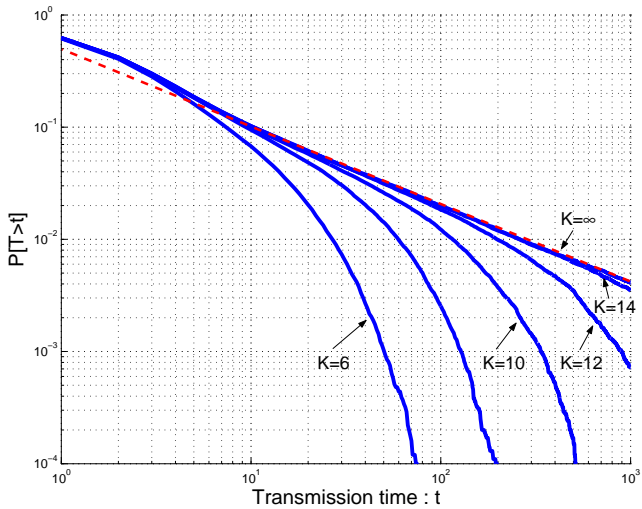}
  \caption{Illustration of the stretched support of the power law main body when the number of users is $\min(M,K)$,
  where $M$ is geometrically distributed.}
\label{fig:finiteUsers}
\end{figure}
We plot the distribution of $\Pr[T>t]$, parameterized by $K$, in
Figure~\ref{fig:finiteUsers}. From the figure we can see that, when
we increase the support of the distributions from $K=6$ to $K=14$,
the main (power law) body of the distribution of $T$ increases from
less than $5$ to almost $700$. This effect is what we call the
stretched support of the main body of $\Pr[T>t]$ in relation to the
support $K$ of $M$. In fact, it can be rigorously shown that the
support of the main body of $\Pr[T>t]$ grows exponentially fast.
Furthermore, it is important to note that, if $K= 14$ and the
probabilities of interest for $\Pr[T>t]$ are bigger than $1/500$,
then the result of this experiment is basically the same as for
$K=\infty$; see Figure~\ref{fig:finiteUsers}.
\end{example}

%\begin{example}
%In many practical applications, we may have combined effects of both
%variable packets and random number of users. In this example, we
%assume that the backoff follows exponential distribution with mean
%$4$, take $M\geq 1$ to be geometric with mean $2$ for the random
%case and $M=2$ for the deterministic case (denoted by $M_r$ and
%$M_d$ respectively), and take packet size $L_r$ to be exponential
%with mean $3$ for the random case and $L_d=3$ for the deterministic
%case (denoted by $L_r$ and $L_d$ respectively).  In order to study
%the impact of the variability of $M$ and $L$ on the transmission
%delay, we combine all the possibilities of $M$ and $L$ to get four
%different scenarios. From the simulation with $10^6$ samples, we see
%that the distribution tail of the combined case is lighter than the
%one with only variable users but heavier than the one with only
%variable packet sizes, implying that the combined effects are
%influenced most by the variability of packet sizes, at least in the
%setting of this particular example.
% \begin{figure}
%\centering
%\includegraphics[width=9cm]{compound.eps}%9
%  \caption{Illustration of the compound effects.}
%\label{fig:compound}
%\end{figure}
%\end{example}

\section{Concluding Remarks and Further Extensions}\label{s:alohaconcluding}
In this paper, we show that a basic finite population ALOHA model
with exponential packets is characterized by power law transmission
delays, possibly even resulting in zero throughput. Based on these
results, we establish a new stability condition that is entirely
derived from the tail behavior of the packet and backoff
distributions.

 Note that at
any moment of time the finite population ALOHA model from
Subsection~\ref{sec:description} can be described as a Markov
process for the state vector $\left(L_1^{(t)}, L_2^{(t)}, \cdots,
L_M^{(t)}\right)$, where $L_i^{(t)}$ is the packet size of user $i$
at time $t$. However, this Markov process is not easy to analyze in
the sense that it has infinitely, possibly uncountably,  many states
with complicated transitions, where long packets tend to accumulate
in the system since the short ones are easier to pass. Hence we
conjecture, based on our initial simulation experiments, that in the
steady state the system may have multiple functional forms for the
power law exponent for different values of $\lambda$, $\nu$ and
$\mu$. The complete characterization of the stability of this Markov
process and the full understanding of the spatial interactions and
temporal correlations of packet sizes in the system remain a
challenging problem. In this paper, we provide a partial picture of
the system behavior. Furthermore, from an engineering perspective,
it is important to study more sophisticated MAC protocols, including
CSMA and RTS/CTS scheme, since ALOHA represents the basis for these
more practical MAC protocols.

This power law effect and the possible instability for our ALOHA
model might be diminished, or perhaps eliminated, by reducing the
variability of packets. However, we show that even a slotted
(synchronized) ALOHA with packets of constant size can exhibit power
law delays when the number of active users is random. This spatial
correlation can have a significant impact on the performance of
ALOHA system when users are persistently present over a period of
time that is larger than the packet transmission time. A more
realistic framework to study this effect could assume that users
arrive and depart on a slower time scale.

 From the algorithmic perspective, we want to point
out that there are other possible ways to reduce the power law
delays in ALOHA, for example, adaptive ALOHA decreases by half the
retransmission rate after each collision, which might greatly reduce
the number of collisions at the expense of possibly low throughput.
 Hence, finding a right balance between the reduction of power law effects and
 a good throughput requires further investigation.

\section{Proofs of Lemmas~\ref{lemma:Nfm} and \ref{lemma:heavyAssmp}} \label{sc:ALOHAproof}

\begin{proof}[{\bf Proof of Lemma~\ref{lemma:Nfm}}] Our proof begins with
finding a subsequence $\mathcal {C}^{(s)}=\{C^{(s)}_{1},
C^{(s)}_{2}, \cdots\}$ from $\mathcal {C}=\{C_{h_m}, C_{h_m+1},
\cdots \}$; recall the definition of $h_m$ preceding the statement
of Lemma~\ref{lemma:Nfm}.  The procedure can be described
iteratively as follows: initially, set $C^{(s)}_{1}=C_{h_m}$, and
for $j\ge 1$, we denote by $C^{(s)}_{j+1}$ the smallest value in
$\mathcal {C}$ that is larger than $C^{(s)}_{j}+1$.  Based on this
subsequence, we define $Y_{j}, j\ge 1$ to be the number of
collisions and departures within each time interval
$\left[C^{(s)}_{j}, C^{(s)}_{j}+1\right]$; note that
$C^{(s)}_{j+1}>C^{(s)}_{j}+1$ by construction. Additionally, let $X$
be the index $j$ of the interval $\left[C^{(s)}_{j},
C^{(s)}_{j}+1\right]$ within which the system reaches the full state
for the first time after $C_{h_m}$, implying $\sum_{j=1}^{X-1}
Y_j\le N^f_m \le \sum_{j=1}^{X} Y_j$. We will prove that there
exists a probability $p>0$, such that for all $j\ge 1$,
$$\Pr\left[X > j\right] \leq (1- p)^j.$$
% To this point, we introduce an
%exponential random variable $V(n), n>0$ with rate $n \nu$, assuming
%$V(0)=\infty$.

To this end, for each interval $\left[C^{(s)}_{j},
C^{(s)}_{j}+1\right], j\ge 1$,
%conditional on
%$\left\{M-U\left(C^{(s)}_{j}+\right) \text{is even} \right\} \cup \{
%X\ge j \}$,
 we construct a special event $\mathcal {E}_j$ such that on this event the system becomes full at a collision
 in  $\left[C^{(s)}_{j}, C^{(s)}_{j}+1\right]$, i.e., there exists $C_l \in \left[C^{(s)}_{j}, C^{(s)}_{j}+1\right]$
 such that $U(C_l+)=M$.   Our construction is described as follows.   We require that all the
backlogged users, including those already in the system immediately
after time $C^{(s)}_{j}$ and the new arrivals in $\left[C^{(s)}_{j},
C^{(s)}_{j}+1/2 \right]$ that collide with other users, make no
retransmissions during the entire interval $\left[C^{(s)}_{j},
C^{(s)}_{j}+1/2 \right]$, which occurs with a probability lower
bounded by $e^{-M\nu/2}$ since there are $M$ users in total and the
backoffs are independent and exponential (memoryless). Also, we
require that all empty users $(\le M)$ observed immediately after
time $C^{(s)}_{j}$ have new arrivals with sizes larger than one
$(L>1)$ within $\left[C^{(s)}_{j}, C^{(s)}_{j}+1/2 \right]$, which
happens with a probability lower bounded by $(1-e^{-\lambda/2})^M
\Pr[L>1]^M$. Now, if $M-U\left(C^{(s)}_{j}+\right)$ is even, our
construction implies that at time $C^{(s)}_{j}+1/2$ all the users
are backlogged, since two consecutive new arrivals after
$C^{(s)}_{j}$ collide with each other and after that they are not
allowed to retransmit before $C^{(s)}_{j}+1/2$, which implies that
there exists $C_{l} \in \left[C^{(s)}_{j}, C^{(s)}_{j}+1/2 \right]$
such that $U\left( C_{l}+\right) =M$.  On the other hand, if
$M-U\left(C^{(s)}_{j}+\right)$ is odd, at time $C_j^{(s)}+1/2$ there
is exactly one user transmitting and the remaining $M-1$ ones are
all backlogged. Now, we require that at least one backlogged user
retransmit during $\left[C^{(s)}_{j}+1/2, C^{(s)}_{j}+1 \right]$,
which occurs with probability $1-e^{-(M-1)\nu /2}$ due to the
memoryless property of the backoff distribution. Clearly, this
requirement ensures that the system is full at a collision time
within $\left[C^{(s)}_{j}+1/2,C^{(s)}_{j}+1\right]$.
 Thus, irrespective of whether $M-U\left(C^{(s)}_{j}+\right)$ is
even or odd, there exists $C_{l} \in \left[C^{(s)}_{j},
C^{(s)}_{j}+1 \right]$ such that $U\left( C_{l}+\right) =M$ on
$\mathcal {E}_j$.

Therefore,  we can uniformly lower bound the probability of
$\mathcal {E}_j$ conditional on $U\left(C^{(s)}_{j}+\right)$, or
equivalently on $M-U\left(C^{(s)}_{j}+\right) $, almost surely
(a.s.) as
\begin{align}\label{eq:expB10}
  \Pr &\left[ \mathcal {E}_j \;{\Big | }\; U\left(C^{(s)}_{j}+\right)
     \right] \nonumber\\
  &\;\; \geq e^{-M\mu/2}  (1-e^{-\lambda/2})^M \Pr[L>1]^M
       \left(1-e^{-(M-1)\nu /2}\right) \nonumber\\
        &\;\; \eqdef p>0.
\end{align}

Now, observe that $\mathcal {E}_j$ is determined by the value of
$U\left( C_j^{(s)}+\right)$ and the future new arrivals and backoff
times after time $C_j^{(s)}$. Furthermore, $\{X\geq j \}$ is
completely determined by the arrival and backoff processes before
time $C_j^{(s)}$.  Hence, due to the memoryless property of the
backoff and arrival times, the event $\{X\geq j\}$ and $\mathcal
{E}_j$ are conditioinally independent given $U\left(
C_j^{(s)}+\right)$.  Therefore,  by using (\ref{eq:expB10}), we
obtain, a.s.,
\begin{align*}
 \Pr &\left[ X \geq j, \mathcal {E}_j \;{\Big | }\; U\left(C^{(s)}_{j}+\right)
 \right] \nonumber\\
 & \quad \;\; =\Pr\left[X \geq j \;{\Big | }\;  U\left(C^{(s)}_{j}+\right)\right]
     \Pr \left[\mathcal {E}_j \;{\Big | }\; U\left(C^{(s)}_{j}+\right)
     \right]\nonumber\\
     &\quad \;\; \geq \Pr\left[X \geq j \;{\Big | }\;  U\left(C^{(s)}_{j}+\right)
     \right] p  ,
 \end{align*}
which implies
\begin{align}\label{eq:expB11}
 \Pr &\left[ X \geq j,  \mathcal {E}^{}_j
 \right]\geq \Pr\left[X \geq j \right]p.
 \end{align}
Thus, by noting that $\{X\geq j\}\cap \mathcal {E}^{}_j \subset
\{X=j\}$ and using (\ref{eq:expB11}), we obtain
 \begin{align*}
\Pr \left[ X \geq j \right] - \Pr \left[ X \geq j+1 \right]&= \Pr
\left[ X =j \right] \nonumber\\
 & \geq \Pr \left[ X \geq j, \mathcal {E}^{}_j \right] \nonumber\\
 &\geq  \Pr\left[X\geq j\right]    p,
 \end{align*}
which results in
\begin{equation*}
    \Pr\left[ X \geq j+1\right]
  \leq  \Pr\left[X\geq j \right] (1-p).
\end{equation*}
Iterating on $j$ in the preceding inequality yields
\begin{equation}\label{eq:expB3}
\Pr[X \geq j+1] \leq (1-p)^{j}.
\end{equation}

 Since the number of collisions and departures $Y_j$ within the interval $\left[C^{(s)}_{j}, C^{(s)}_{j}+1\right]$
 is bounded by the number of active users $U\left(C^{(s)}_{j}+\right)$ in the system immediately after
  time $C^{(s)}_{j}$ plus the total number of retransmissions and arrivals $Z_j$ within this interval,
 we obtain
\begin{equation}\label{eq:expB51}
   Y_j\leq U\left(C^{(s)}_{j}+\right) +Z_j \leq M + Z_j.
\end{equation}
Then, by noting that
  $Z_j$ is
stochastically smaller than a Poisson random variable with rate
$M\max(\nu, \lambda)$ and using (\ref{eq:expB3}), (\ref{eq:expB51}),
we obtain
\begin{align}
 \Pr[N^f_m>n]&\leq \Pr\left[  \sum_{j=1}^{X} Y_j >n \right]
 \nonumber\\
 & \leq \Pr[X>\sqrt{n}] + \sqrt{n} \Pr[Y_1>\sqrt{n}] \nonumber\\
 &=O\left(e^{-\eta \sqrt{n}} \right), \nonumber
\end{align}
for some $\eta>0$,  which finishes the proof.
\end{proof}

\begin{proof}[{\bf Proof of Lemma~\ref{lemma:heavyAssmp}}] Recall the
definition of $\{D_m\}_{m\ge 0}$ in Subsection~\ref{sec:description}
and denote the packet size of the new arrival at time $D_m$ by $L_m$
for all $m$.  First, we prove the case when $M=2$.   For
$m>w_1\eqdef \lceil 1/(\Pr[L>y]) \rceil$, we consider at time $D_m$
a set of departure points $\{D_{m-w_1}, D_{m-w_1+1}, \cdots, D_m\}$.
By the explanation before Lemma~\ref{lemma:heavyAssmp}, we know that
the system has $w_1+1$ number of arrivals in $\left[D_{m-w_1},
D_m\right]$. Define $\tau_{(2,1)}=\min \left\{j: L_j>y, j\geq m-w_1
\right\}$ and it is easy to see that there exists $y_0$ such that
for all $y>y_0$,
\begin{align}
\Pr[\tau_{(2,1)} < m] &= 1- \Pr[L\leq y]^{w_1} \nonumber\\
           &\geq 1-(1-\Pr[L>y])^{1/\Pr[L>y]} \nonumber\\
               &> 1-2e^{-1}, \nonumber
\end{align}
implying that the event $\mathcal {E}_{2}$,  a packet of size larger
than $y$ arriving to the system in $\left[D_{m-w_1},
D_{m-1}\right]$, has a positive probability.
  Now, since
$M=2$, denote  by $L^{(1)}_j$ the size of the packet that arrives
before $D_j$ but is still in the system observed immediately after
$D_j$. Then, define $\mathcal {D}_{2} \eqdef \{L^{(1)}_j>y,
\tau_{(2,1)} < j \leq m \}$, i.e., after packet $\tau_{(2,1)}$
arrives to the system in $\left[D_{m-w_1}, D_{m}\right]$, the size
of the remaining packet in the system observed immediately after
departure times is always larger than $y$. Now, we need to show that
$\Pr[\mathcal{D}_{2} | \mathcal{E}_{2}]$ is also positive. To this
end, we observe that at each point when a departure occurs the new
arrival to the system has a packet size equal in distribution to
$L$, and thus, $\Pr[\mathcal{D}_{2} | \mathcal{E}_{2}]$ is lower
bounded
\begin{align}
&\geq \expect\left[ \prod_{i=\tau_{(2,1)}+1}^{m}\left(\ind(L_i>y) +
\frac{e^{-\nu L_i }}{e^{-\nu L_i}+e^{-\nu y}} \ind(L_i\leq y)
\right)\right], \nonumber
\end{align}
where $\ind(L_i>y) + e^{-\nu L_i }/(e^{-\nu L_i}+e^{-\nu y})
\ind(L_i\leq y)$ gives the lower bound for the probability that
$\underline{L}(D_i)$ is larger than  $y$ at $D_i$. Since $\{L_i\}$
are i.i.d. random variables,  we obtain
\begin{align}\label{eq:steadylowRep}
\Pr[\mathcal{D}_{2} | \mathcal{E}_{2}] &\geq \left(1 -
\expect\left[\frac{e^{-\nu y }}{e^{-\nu L_i}+e^{-\nu y}}
\ind(L_i\leq y)\right]\right)^{w_1}.
\end{align}
Now,  it is straightforward to see that
%\begin{align}
%\expect\left[\frac{e^{-\nu y }}{e^{-\nu L_i}+e^{-\nu y}}
%\ind(L_i\leq y)\right]&\sim \int_{0}^{y}\frac{e^{-\nu y }}{e^{-\nu
%x}+e^{-\nu y}}\mu e^{-\mu x} dx\nonumber \\
%&=\mu e^{-\mu y} \int_{1}^{e^y} \frac{z^{\mu-1}}{1+z^{\nu}}dz
%\nonumber\\
%&\sim c e^{-\mu y}.
%\end{align}
\begin{align}\label{eq:steadylowI12}
\expect&\left[\frac{e^{-\nu y }}{e^{-\nu L_i}+e^{-\nu y}}
\ind(L_i\leq y)\right]= \int_{0}^{y}\frac{1}{1+e^{\nu
(y-x)}}d\Pr[L\leq x]\nonumber \\
&=-\frac{1}{1+e^{\nu(y-x)}}\Pr[L>x] {\Big |}_{0}^{y} +
\nu\int_{0}^{y}
\frac{\Pr[L>x]e^{\nu(y-x)}dx}{(1+e^{\nu(y-x)})^2}\nonumber \\
&= \frac{1}{1+e^{\nu y}} -\frac{1}{2}\Pr[L>y] +
\nu\int_{0}^{\epsilon y}
\frac{\Pr[L>x]e^{\nu(y-x)}dx}{(1+e^{\nu(y-x)})^2}
\nonumber\\
&\quad +\nu\int_{\epsilon y}^{y}
\frac{\Pr[L>x]e^{\nu(y-x)}dx}{(1+e^{\nu(y-x)})^2}.
\end{align}
Since $\nu>\mu$, in the preceding  equality, by choosing
$0<\epsilon<1-\mu/\nu$, we obtain
\begin{align}\label{eq:steadylowI1}
\nu\int_{0}^{\epsilon y}
\frac{\Pr[L>x]e^{\nu(y-x)}dx}{(1+e^{\nu(y-x)})^2}&\leq \nu e^{-\nu
(1-\epsilon) y }=o(\Pr[L>y]).
\end{align}
Next, observe
\begin{align}
&\nu\int_{\epsilon y}^{y}
\frac{\Pr[L>x]e^{\nu(y-x)}dx}{(1+e^{\nu(y-x)})^2} \nonumber\\
&\quad \quad \leq
\nu\int_{\epsilon y}^{y}\Pr[L>x] e^{-\nu (y-x) }dx\nonumber\\
&\quad \quad \leq \nu \Pr[L> y]\int_{\epsilon
y}^{y}\frac{\Pr[L>x]}{\Pr[L>y]} e^{-\nu (y-x) }dx, \nonumber
\end{align}
which, by recalling condition (\ref{eq:steadyCondition}), implies
that there exist $0<\delta<\nu-\mu$ and $y_{\delta}>0$ such that for
all $y>y_{\delta}$,
\begin{align}\label{eq:steadylowI2}
&\nu\int_{\epsilon y}^{y}
\frac{\Pr[L>x]e^{\nu(y-x)}dx}{(1+e^{\nu(y-x)})^2} \nonumber\\
&\quad \quad\leq \nu \Pr[L> y]\int_{\epsilon
y}^{y}e^{(\mu+\delta)(y-x)} e^{-\nu (y-x)}dx
\nonumber\\
&\quad \quad=O(\Pr[L>y]).
\end{align}
Substituting (\ref{eq:steadylowI12}), (\ref{eq:steadylowI1}) and
(\ref{eq:steadylowI2}) into (\ref{eq:steadylowRep}), we obtain, for
$p_2>0$ and $y$ big enough,
\begin{equation}
\Pr[\mathcal{D}_{2} | \mathcal{E}_{2}]>p_2,
\end{equation}
which finishes the proof of the lemma for $M=2$.

Now, we prove the case when $M=3$.  For $m>w_2=\lceil 2/(\Pr[L>y])
\rceil$, consider a set of time points $\mathcal{W}=\{D_{m-w_2},
D_{m-w_2+1}, \cdots, D_m\}$.  Define $\tau_{(3,1)}=\min \left\{j:
L_j>y, j\geq m-w_2 \right\}$ and $\tau_{(3,2)}=\min \left\{j: L_j>y,
j> \tau_{(3,1)}\right\}$. It is easy to see that there exists $y_0$
such that for all $y>y_0$,
\begin{align}
\Pr[\tau_{(3,2)}< m] &\geq \Pr\left[\tau_{(3,1)} \leq \frac{w}{2},
\tau_{(3,2)} -
\tau_{(3,1)}\leq \frac{w}{2}\right] \nonumber\\
               &\geq \left(1- \Pr[L\leq y]^{\frac{w}{2}}\right) \left(1- \Pr[L\leq y]^{\frac{w}{2}}\right) \nonumber\\
               &\geq \left(1-2e^{-1} \right)^2, \nonumber
\end{align}
implying that the event $\mathcal {E}_{3}$,  two packets of size
larger than $y$ arriving to the system in $\left[D_{m-w_2},
D_{m-1}\right]$, has a positive probability.

 Since
$M=3$, we can denote by $L^{(1)}_j \geq L^{(2)}_j$ the order
statistics of the sizes of the packets excluding the one just
arriving to the system at time $D_j$. Then, we define an event
$\mathcal {D}_{3}$ by $\mathcal {D}_{3} \eqdef \left\{ L^{(1)}_j>y,
\tau_{(3,1)} < j \leq \tau_{(3,2)} \right\}$ $\bigcup \left\{
L^{(2)}_j >y, \tau_{(3,2)} < j \leq m \right\}$, i.e., after packet
$\tau_{(3,1)}$ arrives and before packet $\tau_{(3,2)}$ comes to the
system in $\left[D_{m-w_2}, D_{m-2}\right]$, one of the remaining
packet sizes in the system observed immediately after each departure
time is always larger than $y$; after packet $\tau_{(3,2)}$ arrives
to the system in $\left[D_{m-w_2}, D_{m-1}\right]$, all the
remaining packet sizes in the system observed immediately after
departure times are always larger than $y$. Now, we need to show
that $\Pr[\mathcal{D}_{3} | \mathcal{E}_{3}]$ is positive. To this
end, we observe that $\Pr[\mathcal{D}_{3} | \mathcal{E}_{3}]$ is
lower bounded by
\begin{align}
 &\expect {\Bigg [}
\prod_{i=\tau_{(3,1)}+1}^{\tau_{(3,2)}}{\Bigg (} \ind {\Bigg(}\{L_i>
y\} \cup \{L^{(2)}_i> y \} {\Bigg)} \nonumber\\
&\quad \quad  + \frac{e^{-2\nu L_i }+e^{-2\nu L^{(2)}_i}}{e^{-2\nu
L_i}+e^{-2\nu y} +e^{-2\nu L^{(2)}_j}} \ind\left(L_i\leq y,
L^{(2)}_i\leq y\right)
  {\Bigg )}
\nonumber\\
%&\quad \quad \quad \quad \;\;\;\; + \frac{e^{-2\nu L_i }}{e^{-2\nu
%L_i}+e^{-2\nu y}}
%\ind\left(L_i\leq y, L_{(i,2)}^e=0\right) {\Bigg )} \nonumber\\
 & \;\;\;\;\,
 \prod_{i=\tau_{(3,2)+1}}^{w_2} \left(\ind(L_i>
y) + \frac{e^{-2\nu L_i }}{e^{-2\nu L_i}+2e^{-2\nu y}} \ind(L_i\leq
y) \right) {\Bigg ]}, \nonumber
\end{align}
which, by recalling that $\{L_i\}$ are i.i.d. and noting that
$\tau_{(3,2)}-\tau_{(3,1)}\leq w_2$, $w_2-\tau_{(3,2)}\leq w_2$,
implies that $\Pr[\mathcal{D}_{3} | \mathcal{E}_{3}]$ is lower
bounded by
\begin{align}
 & \left( 1-\expect \left[\frac{e^{-2\nu y }}{e^{-2\nu
L_i}+e^{-2\nu y}} \ind(L_i\leq y, L^{(2)}_i \leq y) \right]
\right)^{w_2}  \nonumber \\
&\quad \quad  \times \left(1- \expect\left[\frac{2 e^{-2\nu y
}}{e^{-2\nu L_i}+2e^{-2\nu
y}} \ind(L_i\leq y) \right] \right)^{w_2} \nonumber\\
&\quad \geq \left( 1- \expect \left[\frac{2e^{-2\nu y }}{e^{-2\nu
L_i}+e^{-2\nu y}} \ind(L_i\leq y) \right] \right)^{2w_2}.
\end{align}
Then, by using the same approach as in evaluating
(\ref{eq:steadylowRep}), we can easily obtain, for $p_3>0$,
\begin{equation}
\Pr[\mathcal{D}_{3} | \mathcal{E}_{3}]>p_3,
\end{equation}
which finishes the proof of the case $M=3$.

The situation $M>3$, although notationally complicated,  follows
easily by induction using the same arguments as in proving $M=2, 3$.
For these reasons we omit the details.   \end{proof}

 \small
\bibliographystyle{abbrv}
\bibliography{arXivALOHA}

%
%\begin{IEEEbiography}{Predrag R. Jelenkovi\'c}
%\end{IEEEbiography}
%
%% if you will not have a photo at all:
%\begin{IEEEbiographynophoto}{Jian Tan}
%\end{IEEEbiographynophoto}

% insert where needed to balance the two columns on the last page with
% biographies
%\newpage

% that's all folks
\end{document}